\documentclass[letterpaper,11pt]{article}

\usepackage{hyperref}
\usepackage{graphicx}
\hypersetup{colorlinks=true,citecolor=blue,linkcolor=blue,urlcolor=blue}

\usepackage[UKenglish]{babel}
\usepackage[UKenglish]{isodate}

\usepackage[margin=1in]{geometry}%

\usepackage[title]{appendix}

\usepackage[backgroundcolor=green]{todonotes}

\usepackage[noadjust]{cite}

\usepackage{amsmath}
\usepackage{amssymb}
\usepackage{amsthm}   
\usepackage{enumerate}

\usepackage{tikz}
\usetikzlibrary{shapes.misc, positioning}

\theoremstyle{plain}
\newtheorem{thm}{Theorem}
\newtheorem*{thm*}{Theorem}
\newtheorem{cor}[thm]{Corollary}
\newtheorem{lem}[thm]{Lemma}
\newtheorem{obs}[thm]{Observation}
\newtheorem{conj}[thm]{Conjecture}

\theoremstyle{definition}
\newtheorem{defn}[thm]{Definition}

\newtheorem*{leminterface}{Lemma~\ref{lem:interface}}
\newtheorem*{thmmainbounded}{Theorem~\ref{thm:mainbounded}}
\newtheorem*{thmDCQsubmod}{Theorem~\ref{thm:DCQsubmod}}

\newcommand{\Zinvy}{{\v{Z}}ivn{\'y}}

\mathchardef\hyph="2D
 \let\phi=\varphi

 \def\boldB{\mathbf{B}} 
 \def\boldf{\mathbf{f}} 

 \def\calA{\mathcal{A}}
\def\calB{\mathcal{B}} 

\def\calD{\mathcal{D}} 

\def\calL{\mathcal{L}}

\def\calS{\mathcal{S}}
\def\calT{\mathcal{T}}

\def\numA{\#A}
\def\uhg{H} 
 
\def\sig{\mathrm{sig}}
\def\arity{\mathrm{ar}}
 
\def\homalg{\ensuremath{\textsc{HomAlg}}}
 
\def\bigO{\mathrm{O}}
\def\littleo{\mathrm{o}}
\def\poly{\mathrm{poly}}

\def\comp{\mathrm{comp}}
\def\var{\mathrm{var}} 
\def\proj{\mathrm{proj}}

\let\epsilon=\varepsilon
\let\eps=\epsilon

\newcommand*\from{\colon}

\def\lab{\psi}

\newcommand{\decoracle}{\mbox{\ensuremath{\textsc{Hom}}}}
\def\LIHom{\textsc{LIHom}}
\def\ThreeSAT{$3\textsc{-SAT}$}

\newcommand{\homs}[2]{\mbox{\ensuremath{\mathrm{Hom}(#1 \to #2)}}}

\def\blanksol{\mathsf{Sol}}
\newcommand{\ans}[1]{\mbox{\ensuremath{\mathsf{Ans}(#1)}}}
\newcommand{\sol}[1]{\mbox{\ensuremath{\blanksol(#1)}}}
\def\Ans{\ans}

\def\EdgeFree{\mathrm{EdgeFree}}
\def\Trees{\mathrm{Trees}}

\newcommand{\structsize}[1]{\lVert #1 \rVert}
\newcommand{\vars}[1]{\mathrm{vars}\left(#1\right)}
\newcommand{\free}[1]{\mathrm{free}\left(#1\right)}

\newcommand{\tw}[1]{\mathrm{tw}(#1) }
\newcommand{\cn}[1]{\mathrm{cn}(#1)} 
\newcommand{\fcn}[1]{\mathrm{fcn}(#1)}
\newcommand{\hw}[1]{\mathrm{hw}(#1)}

\newcommand{\fhw}[1]{\mathrm{fhw}(#1)}
\newcommand{\aw}[1]{\mathrm{aw}(#1)}

\newcommand{\CQ}[1]{\mathrm{\#CQ}\left(#1\right)} 
\newcommand{\ECQ}[1]{\mathrm{\#ECQ}\left(#1\right)}
\newcommand{\DCQ}[1]{\mathrm{\#DCQ}\left(#1\right)}
\newcommand{\TA}{\mathrm{\#TA}}

\newcommand{\NP}{\mathrm{NP}}  
\newcommand{\RP}{\mathrm{RP}} 

\title{Approximately Counting Answers to Conjunctive Queries with Disequalities and Negations\thanks{The research leading to these results has received funding from the European Research Council (ERC) under the European Union's Horizon 2020 research and innovation programme (grant agreement No 714532). The research was also supported by the European Research Council (ERC) consolidator grant No 725978 SYSTEMATICGRAPH. Stanislav {\v{Z}}ivn{\'y}\ was supported by a Royal Society University Research Fellowship. The paper reflects only the authors' views and not the views of the ERC or the European Commission. The European Union is not liable for any use that may be made of the information contained therein.}}

\author{Jacob Focke\thanks{CISPA Helmholtz Center for Information Security} \and Leslie Ann Goldberg\thanks{Department of Computer Science, University of Oxford} \and Marc Roth\thanks{School of Electronic Engineering and Computer Science, Queen Mary University of London} \and Stanislav {\v{Z}}ivn{\'y}$^\ddagger$}

\date{4th March 2024\vspace{-5mm}}
\begin{document}
	
\maketitle
\thispagestyle{empty}

\begin{abstract}
We study the complexity of approximating the number of answers to a small query~$\varphi$ in a large database $\calD$. We establish an exhaustive classification into tractable and intractable cases if $\varphi$ is   a conjunctive query possibly including disequalities and negations:
\begin{itemize}
\item If there is a constant bound on the arity of $\varphi$, and if the randomised Exponential Time Hypothesis (rETH) holds, then the problem has a fixed-parameter tractable approximation scheme (FPTRAS) if and only if the treewidth of $\varphi$ is bounded.
\item If the arity is unbounded and $\varphi$ does not have negations, then the problem has an FPTRAS if and only if the adaptive width of $\varphi$ (a width measure strictly more general than treewidth) is bounded; the lower bound relies on the rETH as well.
    \end{itemize}
Additionally  we show that our results cannot be strengthened to achieve a fully polynomial randomised approximation scheme (FPRAS): We observe that, unless $\NP =\RP$, there is no FPRAS even if the treewidth (and the adaptive width) is $1$.
    
However, if there are neither disequalities nor negations, we prove the existence of an FPRAS for queries of bounded fractional hypertreewidth, strictly generalising the recently established FPRAS for conjunctive queries with bounded hypertreewidth due to Arenas, Croquevielle, Jayaram and Riveros (STOC 2021).
\end{abstract}

 \clearpage
 \setcounter{page}{1}

\section{Introduction}\label{sec:Intro}
The evaluation of conjunctive queries  
is amongst the most central and well-studied problems in database theory~\cite{ChandraM77,AbiteboulHV95,ArenasNew, Gottlob02:jcss-hypertree}. 
These queries are also called 
\emph{select-project-join queries}
in relational algebra and \emph{select-from-where} queries in SQL.
In this work, we study the \emph{counting} problem
associated with conjunctive queries and with extensions to conjunctive queries allowing negations, disequalities,
and unions of queries. Given a query $\varphi$ and a database $\calD$, an ``answer'' of~$\phi$ in~$\calD$ is an assignment 
of values from the universe of~$\calD$
to the free variables of~$\varphi$ that can be extended 
(by also assigning values to the existential variables of~$\phi$)
to 
an assignment satisfying~$\phi$.
For example, the universe of the database $\calD$ could be a set of people $U$, and $\calD$ has an entry $F(a,b)$ whenever two people $a,b\in U$ are ``friends''. Then an answer to the query
\begin{equation}\label{eq:friends_query}
    \phi(x) = \exists y\exists z\, F(x,y)\land F(x,z)\land (y\neq z)
\end{equation}
is a person that has at least two friends (from the people in $U$).

The counting problem  is to compute the number of answers of $\varphi$ in $\calD$. 
Our goal is to determine the parameterised complexity of this counting problem
in the situation where   the query $\varphi$ is   significantly smaller than the database $\calD$; a formal exposition is given  in Section~\ref{sec:intro_background}.

Previous work~\cite{DurandM15,ChenM16,DellRW19} established that the problem of exactly counting answers to conjunctive queries  
is \emph{extremely} difficult: Even very simple queries, such as acyclic conjunctive queries, which can be evaluated in polynomial time~\cite{Yannakakis81,GottlobLS01},  
are sufficiently powerful to encode intractable problems in their counting versions, making any non-trivial improvement over the brute-force algorithm impossible under the Strong Exponential Time Hypothesis~\cite{DellRW19}.

Therefore, the relaxation to approximate counting is necessary if efficient algorithms are sought. In this work, we 
quantify the complexity of  approximating the number of answers to
conjunctive queries
with negations and disequalities, and unions thereof, in terms of several 
natural width measures of the queries, such as treewidth, fractional hypertreewidth, and adaptive width. This leads to a complete classification 
(and a new approximation algorithm) in the bounded-arity case, to a complete classification (and another new approximation algorithm) in the unbounded-arity case when negations are excluded,
and to a new FPRAS in the unparameterised setting.
The formal setup, including the definitions of the problems and 
the approximation schemes, are introduced in Section~\ref{sec:intro_background} and we present our results in Section~\ref{sec:results}.

\subsection{Technical Background}\label{sec:intro_background} 
A \emph{signature} $\sigma$ consists of  a finite set of
relation symbols  
with specified positive  arities. 
A \emph{(relational) database} $\calD$ with signature $\sig(\calD)$ consists of 
a finite \emph{universe} $U(\calD)$\footnote{$U(\calD)$ is also often referred to as the ``domain'' of $\calD$.}
   together with, for each 
relation symbol $R\in \sig(\calD)$,
a   relation 
$R^{\calD}$ over  the universe~$U(\calD)$ with the same arity that~$\sig(\calD)$
specifies for~$R$.
The tuples in the relations $R^{\calD}$ are called the \emph{facts} of~$\calD$.
A \emph{conjunctive query} (CQ) $\phi$ with signature $\sig(\phi)$
is a formula of the form 
\begin{equation*} 
	\phi(x_1,\ldots,x_\ell) =  \exists x_{\ell+1} \cdots \exists x_{\ell+k}\psi(x_1,\dots,x_{\ell+k}),
\end{equation*}
where  $\vars{\phi}$  denotes the set of variables $\{x_1, \ldots, x_{k+\ell}\}$ of $\phi$, $\free{\phi}$ denotes the set of free (output) variables $\{x_1,\ldots,x_\ell\}$ of~$\phi$, and  $\psi$ is a conjunction of a finite number of atoms, which 
are predicates of the form 
$R(y_1,\ldots,y_j)$, where $R$ is an arity-$j$ relation symbol in~$\sig(\phi)$
and each $y_i$ is   a variable
in~$\vars{\phi}$.
Each symbol of~$\sig(\phi)$ appears in at least one predicate.
Also, each variable in~$\vars{\phi}$ appears in at least one atom.

In an \emph{extended} conjunctive query (ECQ)
there are three more types of allowable
atoms.
\begin{itemize}
	\item Equality: $y_i = y_j$.
	\item Disequality: $y_i \neq y_j$.
	\item Negated predicate: $\neg R(y_1,\ldots,y_j)$, where 
	$R$ is an arity-$j$ relation symbol in~$\sig(\phi)$
and each $y_i$ is   a variable
in~$\vars{\phi}$. 
\end{itemize}
Again, each element of~$\sig(\phi)$ appears 
at least once in~$\phi$ (as a predicate, as a negated predicate, or both).

In fact, without loss of generality we can assume
that $\phi$ has no equalities, since we can re-write $\phi$ to avoid these by replacing equal variables with a single variable. 
Thus, from now on, we assume that ECQs do not have equalities.

It is natural to extend conjunctive queries by adding disequalities and negations.
Such extended queries were studied for instance in~\cite{Papadimitriou99:JCSS, Koutris17:TCS,Gutierrez-Basulto15:WS, Cima20:AAAI, Arenas11:TCS, Koutris18:PODS}. 
Sometimes we will be interested in extending CQs by adding disequalities but
not negations. We refer to these partially-extended queries as DCQs.

Consider an ECQ~$\phi$.
The following notation of ``Solution'' captures the assignments (of elements in $U(\calD)$ to the variables of~$\phi$) 
that satisfy~$\phi$. It does not distinguish between existential and free variables, but we do that later in Definition~\ref{def:ans}.

\begin{defn} (solution, $\sol{\phi,\calD}$)\label{def:sol}
Let~$\phi$ be an ECQ   and let~$\calD$ be a database with
$\sig(\phi) \subseteq \sig(\calD)$. 
A \emph{solution} of $(\phi,\calD)$ is an 
assignment $\alpha \colon \vars{\phi} \to U(\calD)$ which
has the following property.
\begin{itemize}
\item For every predicate $R(y_1,\ldots,y_j)$ of~$\phi$, 
the tuple $(\alpha(y_1),\ldots,\alpha(y_j))$ is in $R^{\calD}$,  
\item  For every negated predicate $\neg R(y_1,\ldots,y_j)$ of~$\phi$,
the tuple $(\alpha(y_1),\ldots,\alpha(y_j))$ is not in $R^{\calD}$, and
\item For every disequality $y_i \neq y_j$ of~$\phi$ we have $\alpha(y_i)\neq \alpha(y_j)$.
\end{itemize}
We use $\sol{\phi,\calD}$ to denote the set of solutions of $(\phi,\calD)$.  
\end{defn}

In this work, we will not be interested so much in the solutions of $(\phi,\calD)$ 
but rather in their projections onto the
free (output) variables of~$\phi$.

\begin{defn} \label{def:ans} ($\proj$, answer, $\ans{\phi,\calD}$)
Let $\phi$ be an ECQ and  
let $\calD$ be a database with $\sig(\phi) \subseteq \sig(\calD)$. 
Let    $\alpha \colon \vars{\phi} \to U(\calD)$
be an assignment of elements in $U(\calD)$ to the variables of~$\phi$.
We use   $\proj(\alpha,\free{\phi})$ to denote $\alpha$'s projection onto 
the free   variables of~$\phi$. 
That is, $\proj(\alpha,\free{\phi})$
is the assignment from $\free{\phi}$ to 
$U(\calD)$ that agrees with~$\alpha$.
An \emph{answer} 
of $(\phi,\calD)$ is
an assignment $\tau \colon \free{\phi} \to U(\calD)$ 
of elements in $U(\calD)$ to the free variables of~$\phi$
which can be extended to a solution in the sense that
there is a solution~$\alpha$ of
$(\phi,\calD)$ 
with $\proj(\alpha,\free{\phi})=\tau$. We write $\ans{\phi,\calD}$ for the set of all answers of $(\phi,\calD)$.

\end{defn}

Our main focus is on the problem of approximately counting answers to extended conjunctive queries $\phi$, parameterised\footnote{This choice of the parameter allows one to 
produce fine-grained complexity results which
are appropriate for instances in which the query size is (significantly) smaller than the size of the database (see for example the section ``Why Fixed-Parameter Tractability'' in \cite[Section 1]{Marx13:jacm} for a longer discussion of this point). The notion of fixed-parameter tractability is made formal below.} by
the size of~$\phi$, which is denoted by~$\structsize{\phi}$, 
and is defined to be the sum of $|\vars{\phi}|$ and the sum of the arities of the atoms in~$\phi$.

The formal problem definition is as follows.
Let $\Phi$ be a class of ECQs.

\vbox{
\begin{description}\setlength{\itemsep}{0pt}
\setlength{\parskip}{0pt}
\setlength{\parsep}{0pt}   			
\item[\bf Name:] $\ECQ{\Phi}$ 
\item[\bf Input:]   An ECQ $\phi\in \Phi$ and a database $\calD$ with $\sig(\phi) \subseteq \sig(\calD)$.
\item[\bf Parameter:] $\structsize{\phi}$. 
\item[\bf Output:] $|\ans{\phi,\calD}|$. 
\end{description}
}

We define the problems  $\CQ{\Phi}$ and $\DCQ{\Phi}$
analogously, by requiring the input to be a CQ (in the case of $\CQ{\Phi}$) or
a DCQ (in the case of $\DCQ{\Phi}$).

The size of the encoding of the input pair $(\phi,\calD)$ is taken to be the sum of~$\structsize{\phi}$ and
the size of the encoding of $\calD$ (written $\structsize{\calD}$) which is
 defined to be $|\sig(\calD)| + |U(\calD)|$ plus the sum of the lengths of the tuples in the relations of~$\calD$.

Note that singleton unary relations in~$\calD$ can be used to implement ``constants'' in~$\phi$. To see this, for any $v\in U(\calD)$ 
let $R^{\calD}_v$ denote the relation $\{v\}$.
It is possible to refer to the constant~$v$ 
in~$\phi$ by 
including $R_v$ in $\sig(\phi)$
and constraining some variable~$x\in \vars{\phi}$ with the
predicate $R_v(x)$. 
Of course the size $\structsize{\phi}$ increases
by an additional constant amount by the introduction of
the variable~$x$ and the predicate~$R_v(x)$.
Adding all singleton unary relations to the signature of~$\calD$
does not increase~$\structsize{\calD}$ significantly, since 
$|U(\calD)|$ is already included in~$\structsize{\calD}$.

While we focus in this work on counting, 
there is also a corresponding decision problem $\textsc{ECQ}(\Psi)$
with the same input and parameter as $\ECQ{\Psi}$. The output
of $\textsc{ECQ}(\Psi)$ is a single bit,  indicating whether 
$|\ans{\phi,\calD}|>0$.  
The complexity of $\textsc{ECQ}(\Psi)$ is not fully resolved
and some special cases, such as parameterised subgraph isomorphism
are thought to be difficult to resolve~\cite[Chapter 33.1]{DowneyF13}.
We next design the notion of efficient approximation for counting problems such as $\ECQ{\Psi}$.

\paragraph{Randomised Approximation Schemes and Fixed-Parameter Tractability}

Given a value $V$ and $\eps, \delta\in (0,1)$, an $(\eps, \delta)$-\emph{approximation} of $V$ is a random variable $X$ that satisfies $\Pr(|X-V|\le \eps V)\ge 1-\delta$.

Let $\numA$ be a counting problem that, when given input $x$, asks for the value $V(x)$. Slightly overloading notation, an $(\eps, \delta)$-\emph{approximation} for $\numA$ is a randomised algorithm that, given an input $x$ to $\numA$, outputs an $(\eps, \delta)$-approximation of $V(x)$.
A \emph{fully polynomial randomised approximation scheme} (FPRAS) for $\numA$ is a randomised algorithm that, on input $x,\eps,\delta$, computes an $(\eps, \delta)$-approximation of $V(x)$ in time polynomial in $\structsize{x}$, $1/\eps$, and  $\log(1/\delta)$.

Suppose  that the counting problem $\numA$ is parameterised by a parameter $k$
(as the problem $\ECQ{\Phi}$ is parametetrised by~$\structsize{\phi}$). A \emph{fixed-parameter tractable randomised approximation scheme} (FPTRAS) for $\numA$ is a randomised algorithm that, on input $x,\eps,\delta$, computes an $(\eps, \delta)$-approximation of $V(x)$ in time $f(k)\cdot \mathsf{poly}(\structsize{x},1/\eps, \log(1/\delta))$, for some function $f\colon \mathbb{R}\to \mathbb{R}$.

Applying this definition, note that
an FPTRAS for $\ECQ{\Phi}$ has a running time bound of $f(\structsize{\phi})\cdot p(\structsize{\calD},1/\eps, \log(1/\delta))$. In other words, relative to the definition of FPRAS,
the definition of FPTRAS relaxes the condition 
that the algorithm must run in polynomial time  by allowing a super-polynomial factor in the size of the query. 
Since the query is assumed to be significantly smaller than the database, this is a very natural notion of  an efficient algorithm. Indeed, we will show that all FPTRASes for $\ECQ{\Phi}$ constructed in this work cannot be upgraded to FPRASes (subject to natural complexity hypotheses). 
The reason that they cannot be upgraded is that, even for very restricted query classes $\Phi$,  there are reductions from $\NP$-hard problems to the problem of producing an FPRAS for~$\ECQ{\Phi}$.

Considering $\ECQ{\Phi}$ as a parameterised problem allows one to interpolate between the classical complexity of the problem, in which no assumptions
are made regarding the size of the input query, and its \emph{data complexity}, in which the input query is 
fixed. 

From the viewpoint of data complexity, 
there is a brute-force polynomial-time algorithm
for counting answers to a query,  
by iterating through all assignments 
of the variables,   roughly in time $\structsize{\calD}^{\bigO(\structsize{\phi})}$.
If the query~$\phi$ is fixed, the running time 
of this brute-force algorithm is  bounded by a polynomial in the input size. In the fixed-parameter setting
the goal is instead to separate the (potentially exponential) running time in the query size from the (polynomial) running time, in the size of the database.

\subsection{Our Results}\label{sec:results}

In order to give an overview of our results, we provide an illustration in Figure~\ref{fig:my_label}.

\begin{figure}[t!]
    \centering
    		\begin{tikzpicture}[scale=1]
		\draw[ultra thick, rounded corners] (-3.75,8) -- (-3.75,0) -- (3.75,0) -- (3.75,8);
		\draw[thick, rounded corners] (-3.75,1) -- (-3.75,4.5) -- (3.75,4.5) -- (3.75,1);
		\node at (0,4.25) {\footnotesize bounded $\mathrm{tw}$ $\Leftrightarrow$ bounded $\mathrm{(f)hw}$ $\Leftrightarrow$ bounded $\mathrm{aw}$};

		\node[draw, ultra thick, rounded rectangle, fill=green, opacity=0.2] at (0,3.15) {\phantom{FPTRAS for $\mathrm{DCQ}$, $\mathrm{ECQ}$ [Theorem~\ref{thm:mainbounded}]}};
		\node at (0,3.15) {FPTRAS for $\mathrm{DCQ}$, $\mathrm{ECQ}$ [Theorem~\ref{thm:mainbounded}]};
		
		\node[draw, ultra thick, rounded rectangle, fill=red, opacity=0.2] at (0,2.15) {\phantom{No FPRAS for $\mathrm{DCQ}$, $\mathrm{ECQ}$  [Obs.\ \ref{obs:LBnoFPRAS}]}};
		\node at (0,2.15) {No FPRAS for $\mathrm{DCQ}$, $\mathrm{ECQ}$  [Obs.\ \ref{obs:LBnoFPRAS}]};
		
		\node[draw, ultra thick, rounded rectangle, fill=green, opacity=0.2] at (0,1.15) {\phantom{FPRAS for $\mathrm{CQ}$~\cite[Thm 3.2]{ArenasNew}}};
		\node at (0,1.15) {FPRAS for $\mathrm{CQ}$~\cite[Thm 3.2]{ArenasNew}};
		
		\node[draw, ultra thick, rounded rectangle, fill=red, opacity=0.2] at (0,7) {\phantom{\shortstack{No FPTRAS for\\ $\mathrm{CQ}$, $\mathrm{DCQ}$, $\mathrm{ECQ}$ [Obs.\ \ref{obs:LBmain}]}}};
		\node at (0,7) {\shortstack{No FPTRAS for\\ $\mathrm{CQ}$, $\mathrm{DCQ}$, $\mathrm{ECQ}$ [Obs.\ \ref{obs:LBmain}]}};

		\node[draw, ultra thick, rounded rectangle, fill=red, opacity=0.2] at (0,5.65) {\phantom{No FPRAS for $\mathrm{CQ}$~\cite[Cor 3.3]{ArenasNew}}};
		\node at (0,5.65) {No FPRAS for $\mathrm{CQ}$~\cite[Cor 3.3]{ArenasNew}};
		
		\node at (0,-0.5) {Bounded Arity};
		
		\begin{scope}[shift={(8,0)}]
		\draw[ultra thick, rounded corners] (-3.75,8) -- (-3.75,0) -- (3.75,0) -- (3.75,8);
		\draw[thick, rounded corners] (-3.45,0) -- (-3.45,1.5) -- (3.45,1.5) -- (3.45,0);
		\node[draw, ultra thick, rounded rectangle, fill=red, opacity=0.2] at (0,0.65) {\phantom{No FPRAS for $\mathrm{DCQ}$, $\mathrm{ECQ}$  [Obs.\ \ref{obs:LBnoFPRAS}]}};
		\node at (0,0.65) {No FPRAS for $\mathrm{DCQ}$, $\mathrm{ECQ}$  [Obs.\ \ref{obs:LBnoFPRAS}]};
		
		\node at (0,1.3) {\footnotesize bounded tw};
		\draw[thick, rounded corners] (-3.55,0) -- (-3.55,3) -- (3.55,3) -- (3.55,0);	
		\node[draw, ultra thick, rounded rectangle, fill=green, opacity=0.2] at (0,2.15) {\phantom{FPRAS for $\mathrm{CQ}$~\cite[Thm 3.2]{ArenasNew}}};
		\node at (0,2.15) {FPRAS for $\mathrm{CQ}$~\cite[Thm 3.2]{ArenasNew}};
		
		\node at (0,2.8) {\footnotesize bounded hw};
		\draw[thick, rounded corners] (-3.65,0) -- (-3.65,4.5) -- (3.65,4.5) -- (3.65,0);	
		\node[draw, ultra thick, rounded rectangle, fill=green, opacity=0.2] at (0,3.65) {\phantom{FPRAS for $\mathrm{CQ}$ [Theorem~\ref{thm:fractionalFPRAS}]}};
		\node at (0,3.65) {FPRAS for $\mathrm{CQ}$ [Theorem~\ref{thm:fractionalFPRAS}]};
		
		\node at (0,4.3) {\footnotesize bounded fhw};
		\draw[thick, rounded corners] (-3.75,1) -- (-3.75,6) -- (3.75,6) -- (3.75,1);
		\node[draw, ultra thick, rounded rectangle, fill=green, opacity=0.2] at (0,5.15) {\phantom{FPTRAS for $\mathrm{DCQ}$, $\mathrm{CQ}$ [Theorem~\ref{thm:submod}]}};
		\node at (0,5.15) {FPTRAS for $\mathrm{DCQ}$, $\mathrm{CQ}$ [Theorem~\ref{thm:submod}]};
		
		\node at (0,5.8) {\footnotesize bounded aw};
		\node[draw, ultra thick, rounded rectangle, fill=red, opacity=0.2] at (0,7) {\phantom{\shortstack{No FPTRAS for\\ $\mathrm{CQ}$, $\mathrm{DCQ}$, $\mathrm{ECQ}$ [Obs.\ \ref{obs:LBsubmod}]}}};
		\node at (0,7) {\shortstack{No FPTRAS for\\ $\mathrm{CQ}$, $\mathrm{DCQ}$, $\mathrm{ECQ}$ [Obs.\ \ref{obs:LBsubmod}]}};
		\node at (0,-0.5) {Unbounded Arity};
		\end{scope}
	\end{tikzpicture}
    \caption{Overview of our results on approximately counting answers to conjunctive queries ($\mathrm{CQ}$s), to conjunctive queries with disequalities ($\mathrm{DCQ}$s), and to conjunctive queries with disequalities and negations ($\mathrm{ECQ}$s). Upper and lower bounds depend on a variety of width measures of the input queries, namely, treewidth (tw), hypertreewidth (hw), fractional hypertreewidth (fhw), and adaptive width (aw). The equivalence of the width parameters in the case of bounded arity is well known; we provide an explicit argument in Observation~\ref{obs:tw_vs_aw}.
    For completeness, we also compare our results to recent work of Arenas et al.~\cite{ArenasNew}. The lower bounds either rely on the assumption that $\NP \neq \RP$  or on the rETH. All referenced theorems and observations are stated in Section~\ref{sec:results}. Note that, while our results complete the picture of the complexity of $\mathrm{CQ}$, $\mathrm{DCQ}$, and $\mathrm{ECQ}$ in the bounded arity case, two questions remain open for the unbounded arity case: Assuming the adaptive width is bounded, does $\mathrm{ECQ}$ have an FPTRAS, and does $\mathrm{CQ}$ have an FPRAS?}
    \label{fig:my_label}
\end{figure}
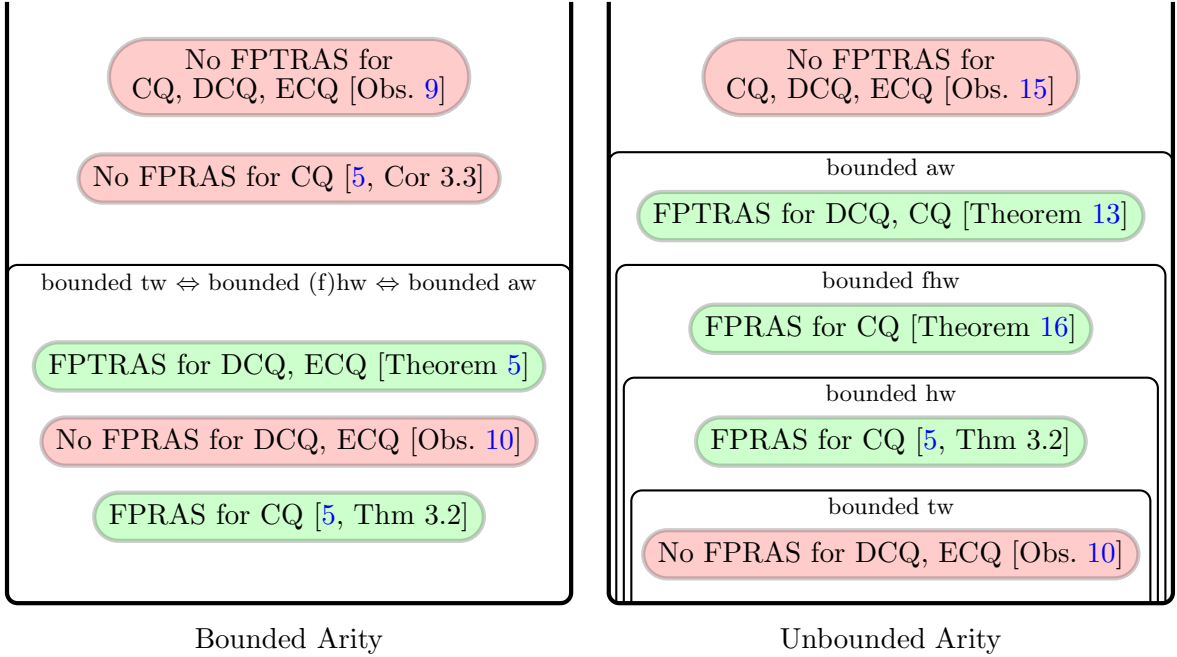

\paragraph{Bounded-Treewidth ECQs}
\label{sec:deftw}

The tractability criteria in our results will depend on the hypergraph associated with a conjunctive query (Definition~\ref{def:Hphi}, see  also~\cite{Gottlob02:jcss-hypertree}).
A \emph{hypergraph} $H$ consists of a (finite) set of vertices~$V(H)$ and a set $E(H) \subseteq 2^{V(H)}$ of non-empty hyperedges. The \emph{arity} of a hypergraph is the maximum size of its hyperedges.

\begin{defn}($H(\phi)$, $\Phi_C$)\label{def:Hphi}
Given an ECQ~$\phi$, the \emph{hypergraph of} $\phi$, denoted $H(\phi)$,
has vertex set $V(H(\phi))= \vars{\phi}$. For each
predicate of~$\phi$ there is a hyperedge in~$E(H(\phi))$
containing the variables appearing in the  predicate.
For each   negated predicate of~$\phi$, there is a hyperedge in~$E(H(\phi))$
containing the variables appearing in the negated predicate.
For any class of hypergraphs $C$, $\Phi_C$  denotes the class of all ECQs $\phi$ with $H(\phi)\in C$.
\end{defn}

We emphasise that~$H(\phi)$ does \textbf{not} contain any hyperedges corresponding to the disequalities of~$\phi$; note that this makes positive results  in terms of $H(\phi)$ stronger, but it also makes these results harder to prove.\footnote{If we included hyperedges corresponding to disequalities, we could just model these disequalities as (binary) relations and reduce to the case of conjunctive queries without disequalities. However, adding those hyperedges can increase the treewidth of the hypergraph (see Definition~\ref{def:tw}) significantly, so it would weaken the results significantly.}

Our first result uses the treewidth of a hypergraph (Definition~\ref{def:tw},  
originally from~\cite{Robertson84:minors3}).  
The definitions of other width measures  
that are used throughout this work, such as fractional hypertreewidth and adaptive width, are deferred to the sections in which they are used.

\begin{defn} (tree decomposition, treewidth)\label{def:tw}
A \emph{tree decomposition} of a hypergraph $H$ is a pair 
$(T,\boldB)$ where $T$ is a (rooted) tree 
and $\boldB$ assigns a subset $B_t \subseteq V(H)$ 
(called a \emph{bag})
to each $t\in V(T)$. The following two conditions are satisfied: (i) for each $e \in
E(H)$ there is a vertex $t \in V(T)$ such that $e \subseteq B_t$, and (ii) for each
$v \in V(H)$ the set $\{ t \in V(T) \mid v \in B_t\}$ induces a (connected)
subtree of $T$.
The treewidth  $\tw{T,\boldB}$ of the tree decomposition $(T,\boldB)$ 
is $\max_{t\in V(T)} |B_t| -1$.
The \emph{treewidth} $\tw{H}$ of $H$ 
is the minimum of $\tw{T,\boldB}$, minimised
 over all tree decompositions $(T,\boldB)$ of~$H$.
\end{defn}

Our first theorem is as follows.

\newcommand{\statethmmainbounded}{\sl Let $t$ and $a$ be positive integers.
Let $C$ be a class of hypergraphs such that every member of $C$ has treewidth at most $t$ and arity at most $a$.  Then $\ECQ{\Phi_C}$ has an FPTRAS, running in time
$\exp(\bigO(||\varphi||^2)) \cdot \poly(\log(1/\delta),\varepsilon^{-1},||\calD||)$.
}

\begin{thm}\label{thm:mainbounded}
\statethmmainbounded
\end{thm}

\paragraph{Technical Challenges}
While, in the case of bounded arity, negated predicates can be simulated by adding a negated relation (symbol) $\overline{R}$ for each relation (symbol) $R$, the disequalities have to be treated separately since we do not include them in the query hypergraph (see the discussion below Definition~\ref{def:Hphi}).

However, the main difficulty stems from the fact that our queries have both quantified and free variables; recall e.g.\ the query in~\eqref{eq:friends_query}. We note that the restricted case in which there are no quantified variables can be dealt with in a much easier way by using standard and well-established reductions from approximate counting to decision. For example, in the special case of arity $2$, and with all disequalities present, approximate counting of answers to queries without quantified variables can be encoded as approximate counting subgraphs, which can be done efficiently for bounded treewidth graphs using colour-coding~\cite{ArvindR02:isaac,Alonetal08}.

If quantified variables are allowed, the situation changes drastically: While it does not matter for the decision version which variables are quantified and which are not, even quantifying a single variable can transform the counting version from easy to infeasible.\footnote{Consider for example the following query in the signature of graphs
\[ \varphi(x_1,\dots,x_k) = \exists y \bigwedge_{i=1}^k E(y,x_i) \,.\] Deciding whether $\varphi$ has an answer is computationally trivial, since it is equivalent to deciding whether there are $k$ (not necessarily distinct) vertices that have a common neighbour. In other words, we can always return ``Yes'' if the graph has at least one edge. However, (exactly) counting answers to $\varphi$ cannot be done in time $O(|V(G)|^{k-\varepsilon})$ unless the Strong Exponential Time Hypothesis fails~\cite{DellRW19}, ruling out any improvement over the brute-force algorithm for the counting version. In the case of approximate counting, the result of Arenas et al.\ \cite{ArenasNew} yields an FPRAS for counting answers to $\varphi$. Also, our result,   Theorem~\ref{thm:mainbounded}, yields an FPTRAS even if we additionally add the constraint that the $x_i$s should be pairwise different. Note that  our result relaxes the  notion of feasibility from an FPRAS to an FPTRAS since  it turns out that the former is not always possible if disequalities are allowed (see Observation~\ref{obs:LBnoFPRAS}).

Finally, we point out that even exact counting becomes easy if we make $y$ a free variable, that is, if we modify the query as follows: $\varphi'(x_1,\dots,x_k,y) = \bigwedge_{i=1}^k E(y,x_i)$. The reason for  this easiness is that counting answers to $\varphi'$ in $G$ is equivalent to counting homomorphisms from $H(\varphi')$ to $G$, which can be done efficiently as $H(\varphi')$ has treewidth~$1$~\cite{DalmauJ04}. } The core technical difficulty arising in both the recent result of Arenas et al.\ \cite{ArenasNew} and also in  our work is the problem of
handling quantified variables in the context of approximate counting. While Arenas et al.\ were able to establish an FPRAS for the case of conjunctive queries of bounded (hyper)treewidth, we show that an FPRAS is not possible if disequalities are allowed (see Observation~\ref{obs:LBnoFPRAS}). For this reason, we relax the condition of feasibility of approximation by aiming for an FP\textbf{T}RAS. 

In this work, we solve the problem of handling quantified variables by relying on the recent $k$-Hypergraph framework of Dell, Lapinskas and Meeks~\cite{DellLM20}, which, in combination with colour-coding, will ultimately establish Theorem~\ref{thm:mainbounded}. We note that the presentation of our methods could be streamlined if we would only consider the bounded arity case. However, since understanding the unbounded arity case is also part of our goal, we present the framework in slightly more generality than we need for the bounded arity case (see Lemma~\ref{lem:interface}).

Before we continue with the presentation of further results, we will present  a useful application of Theorem~\ref{thm:mainbounded}.

\paragraph{Application to Locally Injective Homomorphisms}
Given graphs $G$ and $G'$, a homomorphism from $G$ to $G'$ is a mapping from the vertices of $G$ to the vertices of $G'$ that  maps the edges of $G$ to edges of~$G'$. A homomorphism $h$ is \emph{locally injective} if for each vertex $v$ of $G$, the restriction of $h$ to the neighbourhood $N_{G}(v)$ of~$v$
in~$G$ is injective.
Locally injective homomorphisms have been studied extensively, see~\cite{Fiala08:CSR} for an overview, and they can be applied, for instance, to model interference-free assignments of frequencies (of networks such as wireless networks)~\cite{Fiala01:DAM}. Some recent works on locally injective homomorphisms include~\cite{Rzazewski14:IPL} and~\cite{Bard18:DMTCS}. The complexity of exactly counting locally injective homomorphisms from a fixed graph $G$ to an input graph $G'$ has  been considered and is fully classified~\cite{Roth17:ESA}.

Given $G$ and $G'$, it is easy to construct
an ECQ $\phi(G)$ and a database~$\calD(G')$ such that
locally injective homomorphisms from $G$ to $G'$ are in one-to-one correspondence 
with $\ans{\phi(G),\calD(G')}$.
For this, the signature of $\phi(G)$ and $\calD(G')$ 
has one binary relation~$E$ (representing the edge set of a graph). $\calD(G')$ is the structure representing~$G'$ ---   its universe $U(\calD)$ is $V(G')$ and its relation $E^{\calD(G')}$ is the edge set of~$G'$.
The query~$\phi(G)$ is constructed as follows.
For convenience, let~$k=|V(G)|$ and assume that $V(G)=[k]$.
Let $\cn{G}$ be the set of pairs $i\neq j$ of vertices of $G$ such that $i$ and $j$ have a common neighbour.
Then $\phi(G)$ is the following query~$\phi$ (which has no existential variables).

\[\phi(x_1,\ldots,x_k) = \bigwedge_{\{i,j\}\in E(G)} E(x_i,x_j)~ \wedge ~ \bigwedge_{(i,j) \in \cn{G}} x_i \neq x_j. \]  

Consequently,  Theorem~\ref{thm:mainbounded} also gives an FPTRAS for counting locally injective homomorphisms from graphs $G$ with bounded treewidth. We define the problem of counting locally injective homomorphisms as follows:
Let $C$ and $C'$ be two classes of graphs.

\vbox{
		\begin{description}\setlength{\itemsep}{0pt}
			\setlength{\parskip}{0pt}
			\setlength{\parsep}{0pt}   			
			\item[\bf Name:] $\#\LIHom(C, C')$ 
			\item[\bf Input:]   Graphs $G\in C$ and $G'\in C'$.
			\item[\bf Parameter:] $|V(G)|$. 
			\item[\bf Output:] The number of locally injective homomorphisms from $G$ to $G'$.
		\end{description}
	}
\begin{cor}\label{cor:locallyinj}
	Let $t$ be a positive integer. Let $C_t$ be the class of all graphs with treewidth at most $t$, and let $C$ be the class of all graphs. Then $\#\LIHom(C_t, C)$ has an FPTRAS.
\end{cor}

\paragraph{Matching Lower Bound}
 Theorem \ref{thm:mainbounded} is optimal under standard assumptions ---  subject to the randomised Exponential Time Hypothesis (rETH) there is a matching lower bound showing that there is no FPTRAS 
 for $\ECQ{\Phi_C}$ 
 if the treewidth of~$C$ is unbounded. 
 In fact, there is no FPTRAS for $\CQ{\Phi_C}$ in this case.
 In order to explain the lower bound, we first state the rETH.

\begin{conj}[rETH,~\cite{ImpagliazzoP01}]
There is a positive constant $c$ such that no algorithm, deterministic or randomised, can decide the satisfiability of an $n$-variable
\ThreeSAT\,instance in time $\exp(c\cdot n)$ (with failure probability at most $1/4$).
\end{conj}

The lower bound relies on a result by Marx~\cite[Theorem 1.3]{Marx10:ToC}. We state it here using our notation --- in addition we allow randomised algorithms and therefore replace ETH by rETH.\footnote{In case it is not clear to the reader how Theorem~\ref{thm:canyoubeatTW} matches~\cite[Theorem 1.3]{Marx10:ToC}, we refer to the introduction of~\cite{Marx13:jacm} where this result is stated as Theorem 1.2.}

\begin{thm}[{\cite[Theorem 1.3]{Marx10:ToC}}]\label{thm:canyoubeatTW}
    Let $a$ be positive integer. Let $C$ be a recursively enumerable class of hypergraphs such that every member of $C$ has arity at most $a$. Suppose that the treewidth of hypergraphs in $C$ is unbounded. If there is a computable function $f$ and a randomised algorithm that, given a CQ $\phi\in \Phi_C$ and a database $\calD$ with $\sig(\phi) \subseteq \sig(\calD)$, decides in time $f(H(\phi))\cdot(\structsize{\phi} + \structsize{\calD})^{\littleo(\tw{H(\phi)/\log \tw{H(\phi)}})}$ whether $(\phi, \calD)$ has an answer, then rETH fails.
\end{thm}

In order to apply Theorem~\ref{thm:canyoubeatTW}, note
that an FPTRAS for $\CQ{\Phi_C}$ provides a $(1/2, 1/4)$-approximation
for $\CQ{\Phi_C}$
that runs in time
$f(\structsize{\phi})\cdot \poly(\structsize{\phi}+\structsize{\calD})$.
It follows from the definition
of $\structsize{\phi}$ that $\structsize{\phi}$ is a function of
$H(\phi)$.
Thus, the FPTRAS provides 
a $(1/2, 1/4)$-approximation
for $\CQ{\Phi_C}$ that runs in time
$f(H(\phi))\cdot\poly(\structsize{\phi} + \structsize{\calD})$.
This approximation algorithm can be used to solve the decision problem of determining whether  
the output of $\CQ{\Phi_C}$ is nonzero.
Thus, by Theorem~\ref{thm:canyoubeatTW}, we obtain the following observation.

\begin{obs}\label{obs:LBmain}
Let $a$ be positive integer. Let $C$ be a recursively enumerable class of hypergraphs such that every member of $C$ has arity at most $a$.
If the treewidth of hypergraphs in~$C$ is unbounded then
  $\CQ{\Phi_C}$ does not have an FPTRAS, unless rETH fails.
\end{obs}

Observation~\ref{obs:LBmain} shows that Theorem~\ref{thm:mainbounded} provides a tight result for hypergraph classes with bounded arity in the sense that, assuming rETH, there is an FPTRAS if and only if the treewidth is bounded.
Note that the stated lower bound is slightly stronger than required since it also applies to conjunctive queries without extensions.

Theorem~\ref{thm:mainbounded} 
shows that if $C$ is any class of hypergraphs such that 
every member of~$C$ has treewidth at most~$t$ and arity at most~$a$
then there is an $(\epsilon,\delta)$-approximation algorithm for $\ECQ{\Phi_C}$ that runs in time $f(\structsize{\phi})\cdot\poly(\log(1/\delta),\varepsilon^{-1},||\calD||))$, where $f$ is exponential in $\structsize{\phi}^2$. A natural question is whether the function $f$ can be improved. Specifically, a polynomial $f$ would imply the existence of an FPRAS for 
$\ECQ{\Phi_C}$.
However, the following observation, proved
by reduction from the Hamilton path problem, shows that there is no 
such FPRAS  unless $\NP=\RP$, even when $t=1$ and $a=2$.

\newcommand{\stateLBnoFPRAS}{
\sl 
Let $C$ be the class 
of all hypergraphs with  treewidth at most~$1$ 
and arity at most~$2$. 
Then there is no FPRAS for $\DCQ{\Phi_C}$, unless $\NP =\RP$.
}

\begin{obs}\label{obs:LBnoFPRAS}
\stateLBnoFPRAS 
\end{obs}
\begin{proof}
Given an $n$-vertex graph~$G$, we will show how to construct 
(in time $\poly(n)$)
an instance
$(\phi,\calD)$ 
of $\DCQ{\Phi_C}$ such 
that 
the answers in
$\Ans{\phi,\calD}$ are in one-to-one correspondence
with the   Hamiltonian paths of~$G$.
This implies (e.g., \cite[Theorem 1]{relative})
that there is no FPRAS   for $\DCQ{\Phi_C}$  unless $\NP=\RP$.

The construction is as follows. $U(\calD) = V(G)$.
The signature $\sig(\phi)=\sig(\calD)$ contains a single binary relation symbol~$E$.
The relation $E^{\calD}$ is the edge set~$E(G)$.
The query~$\phi$  
is defined as follows. 
\[
    \phi(x_1,\dots,x_n) = \bigwedge_{i \in [n-1]} E(x_i,x_{i+1}) ~ \wedge ~ \bigwedge_{1\leq i < j \leq n} x_i \neq x_j.
\]

Note that $\phi$ has no existential variables so 
the solutions of  $(\phi,\calD)$ are in one-to-one correspondence
with $\Ans{\phi,\calD}$.
It is clear from the definition of~$\phi$ that  these 
are also in one-to-one correspondence with the Hamilton paths of~$G$.

It remains to show that $\phi\in \Phi_C$, that is that $H(\phi)$ has 
treewidth~$1$ and arity~$2$. Both of these follow from the
fact that $H(\phi)$ is the path $x_1,\ldots,x_n$, which follows
from the definition of $H(\phi)$ (Definition~\ref{def:Hphi}). 
\end{proof}

Interestingly, the situation changes if we consider CQs without extensions.
Arenas, Croquevielle, Jayaram, and Riveros~\cite{ArenasNew} 
give an FPRAS for $\CQ{\Phi_C}$ 
for any bounded-treewidth class~$C$ of hypergraphs. Their result will be stated formally as Theorem~\ref{thm:Arenas} in Section~\ref{sec:fh} (where it will be improved). 
 Observation~\ref{obs:LBnoFPRAS} ensures that their result cannot be
 generalised to 
 cover extended conjunctive queries  (unless $\NP=\RP$). However, their result does apply to classes of hypergraphs with unbounded arity.

\paragraph{Beyond Bounded Arity}

Even though Observation~\ref{obs:LBmain} gives a tight lower bound for classes $C$ of hypergraphs with bounded arity, there is room for improvement if the arity in $C$ is unbounded.
In this setting, it is worth  considering other notions of 
hypergraph width, such as hypertreewidth, fractional hypertreewidth,  adaptive width and submodular width. These width measures will be defined in the sections where they are used. Here we just give the relationships between them, from \cite[Figure 2]{Marx13:jacm}.

\begin{defn}[weakly dominated, strongly dominated, weakly equivalent]
A width measure is a function from hypergraphs to $\mathbb{R}_{\geq 0}$.
Given two width measures $w_1$ and $w_2$, we say that $w_1$ is  \emph{weakly dominated} by $w_2$ 
if there is a function~$f$
such that every hypergraph~$H$
has $w_2(H) \leq f(w_1(H))$. 
We say that $w_1$ is \emph{strongly dominated} by $w_2$
if $w_1$ is weakly dominated by $w_2$ and
there is a class of hypergraphs 
that has unbounded $w_1$-width, but bounded $w_2$-width.
We say that $w_1$ and $w_2$ are \emph{weakly equivalent} if they weakly dominate each other.
\end{defn}

If $w_1$ is strongly dominated by $w_2$ then the class
of hypergraphs with bounded $w_1$-width is strictly contained in 
the class of hypergraphs with bounded $w_2$-width.
Thus, algorithmic results based on bounded $w_2$-width have
strictly greater applicability than algorithmic results based on bounded $w_1$-width. If $w_1$ and $w_2$ are weakly equivalent then algorithmic results for bounded $w_1$-width and bounded $w_2$-width are equivalent.

\begin{lem}[\cite{Marx13:jacm}]\label{lem:dominate}\label{lem:dom}
Treewidth is strongly dominated by hypertreewidth.
Hypertreewidth is strongly dominated by fractional hypertreewidth. Fractional hypertreewidth is strongly dominated by adaptive width, which is weakly equivalent to submodular width.
\end{lem}

Note that all of the width measures 
from Lemma~\ref{lem:dom} are weakly equivalent if we assume an overall bound on the arity of hypergraphs; we make this formal in Observation~\ref{obs:tw_vs_aw}.

When restricting the queries to DCQs instead of ECQs, we 
can improve Theorem~\ref{thm:mainbounded} by allowing unbounded arity, extending it to adaptive width. The following Theorem is proved in Section~\ref{sec:sw}.

\newcommand{\stateDCQsubmod}{\sl Let $b$ be a positive integer.
Let $C$ be a class of hypergraphs such that every member of $C$ has
adaptive width at most $b$. Then $\DCQ{\Phi_C}$ has an FPTRAS.}

 \begin{thm}\label{thm:DCQsubmod} \label{thm:submod}
 \stateDCQsubmod
\end{thm}

\paragraph{Further Technical Challenges} 
In addition to the challenges that arose in the bounded arity case, a further issue that arises in the unbounded arity case is finding the right notion of treewidth on hypergraphs; recall that the notions of treewidth, hypertreewidth, fractional hypertreewidth and adaptive width are all equivalent in the bounded arity setting, but not in the unbounded arity setting. 

We ultimately ended up with adaptive width as the most general width measure for which we can establish the existence of an FPTRAS. Despite the fact that approximate counting is often harder than decision~\cite{relative,BulatovZ20}, we find that in the current setting, the criterion for tractability is the same for decision and approximate counting. In fact, it turns out that Theorem~\ref{thm:DCQsubmod}, i.e., the choice of adaptive width, is optimal, unless the rETH fails. This matching lower bound comes from another result of Marx~\cite[Theorem 7.1]{Marx13:jacm}, which we express using our notation; for what follows, we write $\aw{H}$ for the adaptive width of a hypergraph $H$:

 \begin{thm}[{\cite[Theorem 7.1]{Marx13:jacm}}]\label{thm:canyoubeatSW}
    Let $C$ be a recursively enumerable class of hypergraphs with unbounded   adaptive width. If there is a computable function $f$ and a randomised algorithm that, given a CQ $\phi\in \Phi_C$ and a database $\calD$ with $\sig(\phi) \subseteq \sig(\calD)$, decides in time $f(H(\phi))\cdot(\structsize{\phi} + \structsize{\calD})^{\littleo(\aw{H(\phi)}^{1/4})}$ whether $(\phi, \calD)$ has an answer, then rETH fails.
\end{thm}

As we noted earlier,  an FPTRAS for $\CQ{\Phi_C}$ provides a $(1/2, 1/4)$-approximation
for $\CQ{\Phi_C}$
that runs in time
$f(\structsize{\phi})\cdot \poly(\structsize{\phi}+\structsize{\calD})$.
It follows from the definition
of $\structsize{\phi}$ that $\structsize{\phi}$ is a function of
$H(\phi)$.
Thus, the FPTRAS provides 
a $(1/2, 1/4)$-approximation
for $\CQ{\Phi_C}$ that runs in time
$f(H(\phi))\cdot\poly(\structsize{\phi} + \structsize{\calD})$.
This approximation algorithm can be used to solve the decision problem of determining whether  
the output of $\CQ{\Phi_C}$ is nonzero.
Thus, by Theorem~\ref{thm:canyoubeatSW}, we obtain the following observation.

 \begin{obs}  \label{obs:LBsubmod}
Let $C$ be a recursively enumerable class of hypergraphs with unbounded adaptive width. Then
  $\CQ{\Phi_C}$ does not have an FPTRAS, unless rETH fails.
\end{obs}

 Clearly, Observation~\ref{obs:LBsubmod} also 
 rules out an FPTRAS for 
 $\DCQ{\Phi_C}$ (matching Theorem~\ref{thm:DCQsubmod}) 
 or an FPTRAS for 
 $\ECQ{\Phi_C}$
 when $C$ has unbounded adaptive width.

 As stated in Observation~\ref{obs:LBnoFPRAS}, there is little hope of improving Theorem~\ref{thm:submod} to obtain an FPRAS instead of an FPTRAS. However, as mentioned previously, if disequalities are not part of the query, the result by Arenas et al.~\cite{ArenasNew} also applies to classes of hypergraphs with unbounded arity --- they give an FPRAS for $\CQ{\Phi_C}$ if 
$C$ is a class of hypergraphs with bounded
 hypertreewidth. 
 We  improve this result by showing that it suffices to require bounded fractional hypertreewidth --- the following theorem is proved in Section~\ref{sec:fh}.

\begin{thm}\label{thm:fractionalFPRAS} 
Let $b$ be a positive integer.
Let $C$ be a class of hypergraphs such that every member of $C$ has fractional hypertreewidth at most $b$.
Then $\CQ{\Phi_C}$ has an FPRAS.
\end{thm}

We don't know whether Theorem~\ref{thm:fractionalFPRAS}
can be extended to instances with bounded adaptive width (closing the gap between Theorem~\ref{thm:fractionalFPRAS} and the lower bound presented in Observation~\ref{obs:LBsubmod}). However, this situation now mirrors the situation in the decision world: Bounded fractional hypertreewidth is the most general property of the class of underlying hypergraphs, for which the conjunctive query decision problem is known to be polynomial-time solvable, see~\cite{Marx10:talg}. Fixed-parameter tractability for this problem is also known for classes with bounded adaptive width~\cite{Marx13:jacm}, complemented by the matching lower bound stated in Theorem~\ref{thm:canyoubeatSW}. The question whether bounded fractional hypertreewidth is the right answer for the existence of a polynomial-time algorithm remains open and we refer to the conclusion of~\cite{Marx13:jacm} for further discussion regarding this question.
With our results from Theorems~\ref{thm:submod} and~\ref{thm:fractionalFPRAS} we now arrive at precisely the same gap for the existence of an FPRAS for the corresponding counting problem \#CQ.

\paragraph{Extensions: Sampling and Unions}

Using standard methods, the algorithmic results presented in this work can be extended in two ways. First, approximate counting algorithms can be lifted to obtain algorithms for approximately uniformly sampling answers. This is based on the fact that the problems that we considered  are self-reducible (self-partitionable)~\cite{Schnorr76:icalp,Jerrum86:tcs,Dyer99:lmslns}. Second, instead of single (extended) conjunctive queries, one can also consider unions thereof.
We refer to Section~\ref{sec:extensions} for further details.

\subsection{Related Work}
The first systematic study of the complexity of \emph{exactly} counting answers to conjunctive queries is due to Pichler and Skritek~\cite{PichlerS13}, and Durand and Mengel~\cite{DurandM15}. Their result has been extended to unions of conjunctive queries by Chen and Mengel~\cite{ChenM16} and to extended\footnote{``Disequalities'' are sometimes referred to as ``Inequalities'', and ``Negations'' are sometimes referred to as ``Non-monotone Constraints''.} conjunctive queries by Dell, Roth and Wellnitz~\cite{DellRW19}. Unfortunately, all of those results have shown that exact counting is infeasible even for very restricted classes of queries, indicating that relaxing to approximate counting is necessary if efficient algorithms are sought for wide classes of queries. As highlighted before, a recent result by Arenas et al.~\cite{ArenasNew} shows that 
there is an FPRAS for
approximately counting answers to conjunctive queries whenever the hypergraphs of the queries have bounded hypertreewidth, yielding a significantly wider class of tractable instances than for exact counting. We state and further discuss their result in Section~\ref{sec:fh}.

If we restrict to instances without existential quantifiers then counting answers to conjunctive queries is equivalent to the problem of counting homomorphisms from a small relational structure to a large one, the complexity of which was investigated by Dalmau and Jonsson~\cite{DalmauJ04} in the case of exact counting and by Bulatov and \Zinvy~\cite{BulatovZ20} in the case of approximate counting.

The notion of an FPTRAS was introduced by Arvind and Raman~\cite{ArvindR02:isaac} and has since been established as the standard notion for tractability of parameterised approximate counting problems (see~\cite{Meeks16} for an overview).
 
\subsection{Algorithmic Methods and Proof Techniques}

The FPTRAS presented in our main result, Theorem~\ref{thm:mainbounded}, relies on a framework that was introduced in a recent work by  Dell, Lapinskas, and Meeks~\cite{DellLM20}. Their framework establishes an algorithm for approximating the number of hyperedges in a hypergraph that uses an oracle for a related decision problem. Our contribution is to figure out how to reduce the
problem of approximating answers to the problem of approximately counting hyperedges in 
an appropriately-constructed hypergraph, thus giving an algorithmic result that completely matches
the corresponding hardness result (Observation~\ref{obs:LBmain}).
Our algorithm for the unbounded-arity case (Theorem~\ref{thm:DCQsubmod}) 
has the same general structure and here an additional contribution is
determining the correct criterion for the DCQ case (which turns out to be bounded adaptive width).

The FPRAS that we present in  Theorem~\ref{thm:fractionalFPRAS}
first  constructs a tree decomposition in a convenient format, then   collects appropriate local information in the tree decomposition. Using
this local information, it reduces the problem of approximately counting answers to the
problem of approximately counting outputs accepted by a tree automaton (which can
be accomplished by an algorithm  of Arenas, Croquevielle, Jayaram, and Riveros~\cite{ArenasNew}).
A key observation leading to the improved result is that
the tree automaton reduction still works even when the tree decomposition may have more than polynomially many hyperedges per bag, as long as the number of relevant partial solutions 
is bounded
by a polynomial, as is the case for bounded
fractional hypertreewidth~\cite{Grohe14:talg}.

\section{Technical Preliminaries}\label{sec:Prelims}

\subsection{Using Decision Oracles to Approximately Count Hyperedges}

We start by introducing the terminology that we need.
A hypergraph $H$ is called $\ell$\emph{-uniform} if each hyperedge of $H$ has cardinality~$\ell$.
An $\ell$\emph{-partite subset} of a (finite) set $V$ is a tuple $(V_1,\dots,V_\ell)$ of (pairwise) disjoint subsets of $V$. 
This is called an \emph{$\ell$-partition} of~$V$
if $V = \cup_{i=1}^{\ell} V_i$.

An $\ell$-uniform hypergraph~$H=(V,E)$ is 
$\ell$\emph{-partite} with $\ell$-partition $(V_1,\dots,V_\ell)$ if 
$(V_1,\ldots,V_\ell)$ is an $\ell$-partition of~$V$ and
every hyperedge in~$E$ contains exactly one vertex in each $V_i$. 

Given an $\ell$-uniform hypergraph $H=(V,E)$ and an $\ell$-partite subset $(V_1,\dots,V_\ell)$ of~$V$, the hypergraph $H[V_1,\dots,V_\ell]$ has vertex set  $\bigcup_{i=1}^{\ell} V_i$; recall that the $V_i$ are pairwise disjoint.
The hyperedge set of $H[V_1,\ldots,V_\ell]$ 
is the set of all hyperedges in~$E$  that contain (exactly) one vertex in each $V_i$.
Note that $H[V_1,\dots,V_\ell]$ is $\ell$-partite. 

We write ${\EdgeFree}(H)$ for the predicate that is satisfied if a hypergraph $H$ has no edges. The main result of Dell, Lapinskas and Meeks is as follows.

  \begin{thm}[\cite{DellLM20}]\label{thm:colreduction}
 There is an algorithm $\mathbb{A}(\varepsilon,\delta,H)$ with the following behaviour. Suppose that $0<\varepsilon,\delta<1$ are positive reals, $H$ is an $\ell$-uniform hypergraph, and that (in addition to learning $V(H)$ and $\ell$) the algorithm $\mathbb{A}$ has access to an oracle for evaluating the predicate 
${\EdgeFree}(H[V_1,\ldots,V_\ell])$ for any $\ell$-partite subset $(V_1,\ldots,V_\ell)$ of $V(H)$.

$\mathbb{A}$ computes an $(\varepsilon,\delta)$-approximation of $|E(H)|$ in time $\bigO(NT)$, using at most $T$ calls to the oracle, where $N=|V(H)|$ and $T=\Theta(\log(1/\delta)\varepsilon^{-2} \ell^{6\ell}{(\log N)}^{4\ell+7})$. 
 \end{thm}
 
 We emphasise that, while learning $V(H)$, $\ell$, $\varepsilon$, and $\delta$,  the algorithm in Theorem~\ref{thm:colreduction} does not learn $E(H)$ as part of the input.
 The only access that the algorithm has to~$E(H)$ is via
 the oracle which evaluates the predicate
 ${\EdgeFree}(H[V_1,\ldots,V_\ell])$.
  In fact,  in our application $H$ will be the hypergraph whose hyperedges are the elements of $\ans{\phi,\calD}$ 
 for an ECQ~$\phi$ 
with $\ell$ free variables and a database~$\calD$ with 
$|U(\calD)|=N$. 
 Theorem~\ref{thm:colreduction}  will help us to reduce the problem of approximating $|\ans{\phi,\calD}|$
 to the (decision) problem of determining whether one exists via the oracle for the predicate ${\EdgeFree}$.

\subsection{From Conjunctive Queries to Relational Structures and Homomorphisms}\label{sec:QueriesToStructures}

It is well-known that answers to conjunctive queries are closely related to homomorphisms between relational structures. In this work, it will be convenient to view answers  in the language of homomorphisms.
We start with the relevant definitions.

Recall that 
a {signature} $\sigma$ consists of a finite set of relation symbols with specified positive arities.  
We use $\arity(R)$ to denote the arity of a relation symbol~$R$ and $\arity(\sigma)$ to denote the maximum arity of any relation symbol in~$\sigma$.
Given a signature~$\sigma$,
a \emph{structure}~$\calA$ with signature $\sig(\calA)=\sigma$
consists of a finite universe $U(\calA)$ and, for each relation symbol $R\in \sigma$,
a relation $R^\calA \subseteq U(\calA)^{\arity(R)}$.
Following~\cite{Grohe07:jacm}, we use $\structsize{\calA}$ to denote the 
size of structure~$\calA$, which is given by 
$\structsize{\calA} = |\sig(\calA)| + |U(\calA)| + \sum_{R\in \sig(\calA)} |R^\calA| \cdot \arity(R)$. Note that a (relational) database is a structure.

Given two structures $\calA$ and $\calB$ with 
$\sig(\calA) \subseteq \sig(\calB)$,
a \emph{homomorphism} from $\calA$ to $\calB$ is a function $h\from U(\calA) \to U(\calB)$ such that for all $R\in \sig(\calA)$ with $t=\arity(R)$ and all tuples $(a_1,\ldots,a_t)\in R^{\calA}$ it holds that $(h(a_1),\ldots,h(a_t))\in R^{\calB}$. Then $\homs{\calA}{\calB}$ denotes the set of homomorphisms from $\calA$ to $\calB$.

Let $\phi$ be an ECQ and let $\calD$ be a database 
with $\sig(\phi) \subseteq \sig(\calD)$. We define a pair of {associated structures} $\calA(\phi)$ and $\calB(\phi, \calD)$ with the goal of expressing query answers in~$\ans{\phi,\calD}$ as homomorphisms from~$\calA$ to~$\calB$. 

\begin{defn}[$\calA(\phi)$]\label{def:calH}
    The universe of $\calA(\phi)$ is
    $U(\calA(\phi)) = \vars{\phi}$.
    The signature of $\calA(\phi)$ is 
    constructed from $\sig(\phi)$ as follows
    \begin{itemize}
    \item If there is a predicate in~$\phi$ involving the relation symbol~$R$, then $R$ is in the signature of~$\calA(\phi)$ (with the same arity as $R$ has in $\sig(\phi)$).
    \item If there is a negated predicate in~$\phi$ involving the relation symbol~$R$ then the relation symbol~$\overline{R}$ is in the signature of~$\calA(\phi)$ (with the same arity as $R$ has in $\sig(\phi)$).
    \end{itemize}
    Finally, for each $R\in \sig(\phi)$ and $j=\arity(R)$, $R^{\calA(\phi)}$ is the set of tuples $(y_1,\ldots, y_j)$ for which $R(y_1,\ldots, y_j)$ is a predicate in~$\phi$; and $\overline{R}^{\calA(\phi)}$ is the set of tuples $(y_1,\ldots, y_j)$ for which $\neg R(y_1,\ldots, y_j)$ is a negated predicate in~$\phi$.
\end{defn}

For $R\in \sig(\phi)$, let $P_\phi^+(R)$ be the set of predicates of $\phi$ that use $R$, and let $P_\phi^-(R)$ be the set of negated predicates of $\phi$ that use $R$. Then 
\[
|U(\calA(\phi))| + \sum_{R\in \sig(\calA(\phi))} |R^{\calA(\phi)}| \cdot \arity(R) =|\var(\phi)| + \sum_{R\in \sig(\phi)} \bigl(|P_\phi^+(R)| + |P_\phi^-(R)|\bigr)\cdot\arity(R) \le \structsize{\phi}.
\]

Recall that 
$\structsize{\phi}$ is the sum of  $ |\vars{\phi}|$
  and the sum of the arities of the atoms in~$\phi$.
Also,  
the size of a structure $\calA$ is
$\structsize{\calA} = |\sig(\calA)| + |U(\calA)| + \sum_{R\in \sig(\calA)} |R^\calA| \cdot \arity(R)$.
Thus, we obtain the following observation regarding the size of $\calA(\phi)$.
\begin{obs}\label{obs:sizeOfcalA}\label{obs:A1}
Let~$\phi$ be an ECQ with $\nu$ negated predicates.
Then
$\structsize{\calA(\phi)}\le |\sig(\phi)| +\nu + \structsize{\phi}
\leq 3\structsize{\phi}$.
 \end{obs}

\begin{defn}[$\calB(\phi, \calD)$]\label{def:calG}
    The universe of~$\calB(\phi, \calD)$ is the universe $U(\calD)$.
    The signature is $\sig(\calB(\phi, \calD)) = \sig(\calA(\phi)) $ and the relations are defined as follows.
    \begin{itemize}
    \item For each $R\in \sig(\calA(\phi)) \cap \sig(\calD)$, $R^{\calB(\phi, \calD)} = R^\calD$.
    \item For each $\overline{R} \in
    \sig(\calA(\phi))\setminus \sig(\calD)$,  $\overline{R}^{\calB(\phi, \calD)} = U(\calD)^{\arity(R)} \setminus R^\calD$.
\end{itemize}
\end{defn}

Note that $\structsize{\calB(\phi, \calD)}$ is bounded from above by
\[
    |\sig(\calA(\phi))| + |U(\calD)| + \sum_{R\in \sig(\calA(\phi))\cap\sig(\calD)} |R^\calD| \cdot \arity(R) + \sum_{\overline{R}\in \sig(\calA(\phi))\setminus \sig(\calD)} |U(\calD)|^{\arity(R)} \cdot \arity(R),
\]
where, for an ECQ $\phi$ with $\nu$ negated predicates $|\sig(\calA(\phi))| \le |\sig(\phi)| +\nu \leq |\sig(\calD)| +\nu$.
Thus, we obtain the following observation regarding the size of $\calB(\phi,\calD)$.

\begin{obs}\label{obs:sizeOfcalB}\label{obs:A2}
Let~$\phi$ be an ECQ with
$a=\arity(\sig(\phi))$. If~$\phi$ has  $\nu$ negated predicates  
then 
$\structsize{\calB(\phi,\calD)}\le \structsize{\calD} +\nu + \nu\, a\,|U(\calD)|^{a}
\leq 
2\structsize{\phi}(
 \structsize{\calD} + \nu   {|U(\calD)|}^a)$.
\end{obs}

Let $\phi$ be an ECQ and  
let $\calD$ be a database with $\sig(\phi) \subseteq \sig(\calD)$.
Let $\Delta(\phi)=\{\{x_i,x_j\} \mid x_i\neq x_j \text{ is an atom of } \phi\}$.
Note that 
$\sol{\phi,\calD} = \{  h \in \homs{\calA(\phi)}{\calB(\phi,\calD)}:  
     \forall \{x_i,x_j\}\in \Delta(\phi): h(x_i)\neq h(x_j) \}$.
Therefore, 

\begin{equation}\label{equ:Ans}
\begin{split}
    \ans{\phi, \calD} = \{ \tau\from \free{\phi} \to U(\calD) \mid\, 
    &\exists h \in \homs{\calA(\phi)}{\calB(\phi,\calD)}: \proj(h,{\free{\phi}}) = \tau \\
    &\wedge\, \forall \{x_i,x_j\}\in \Delta(\phi): h(x_i)\neq h(x_j) \}.
\end{split}
\end{equation}

\section{Using Decision Oracles to Count Answers}\label{sec:todec}

We start by defining the homomorphism decision problem.

\vbox{
		\begin{description}\setlength{\itemsep}{0pt}
		\setlength{\parskip}{0pt}
		\setlength{\parsep}{0pt}   			
		        \item[\bf Name:] $\decoracle$ 
			    \item[\bf Input:]   
			    Structures~$\calA$ and~$\calB$ with $\sig(\calA) \subseteq \sig(\calB)$.
			    \item[\bf Output:] Is there a homomorphism from~$\calA$ to~$\calB$? 
		\end{description}
	}

The goal of this section is to establish Lemma~\ref{lem:interface}, the proof of which will be obtained by a combination of the $k$-hypergraph framework (Theorem~\ref{thm:colreduction}) and colour-coding.

\newcommand{\stateleminterface}{\sl
    There is a randomised algorithm   that is equipped with oracle access to $\decoracle$ and takes the following inputs
    \begin{itemize} 
        \item  an ECQ $\phi$,
        \item a database~$\calD$ with $\sig(\phi) \subseteq \sig(\calD)$,
        \item rational numbers~$\epsilon$ and $\delta$ in $(0,1)$.
    \end{itemize} 
    Let $a = \arity(\sig(\phi))$ and let $\nu$ be the number of negated predicates in~$\phi$.
    The algorithm  computes an $(\epsilon,\delta)$-approximation of  $|\ans{\phi,\calD}|$ in time \[
    \exp(\bigO({\structsize{\phi}}^2)) \cdot
    \poly( \log(1/\delta),\epsilon^{-1},   \structsize{\calD},
    { \nu |U(D)|^{a}}  
    ) \,.\]
    Each oracle query $\decoracle(\widehat\calA,{\widehat\calB})$
    that is made by the algorithm has the property that
    $\widehat{\calA}$ can be obtained from $\calA(\phi)$ by adding   unary relations  
and satisfies
$\structsize{\widehat\calA}  \leq 5\structsize{\phi}^2$.
}

\begin{lem}\label{lem:interface} 
    \stateleminterface
\end{lem}

In order to prove Lemma~\ref{lem:interface} we require some prerequisites and definitions. 

\begin{defn} \label{def:Ui} ($U_i(\calD)$)
Given a database $\calD$ and an integer $i$ we define $U_i(\calD)=U(\calD)\times \{i\}$. 
\end{defn}

Intuitively,   $U_i(\calD)$ will be used to specify the image of the $i$th variable of some query $\phi$.

\begin{defn}[$\uhg(\phi,\calD)$]\label{def:l-uniformHypergraph}\label{def:H}
Let $\phi$ be an  ECQ and let $\calD$ be a database
with $\sig(\phi) \subseteq \sig(\calD)$.  
Let $\ell = |\free{\phi}|$ and let  
$x_1, \ldots, x_{\ell}$ be an enumeration of 
the variables in~$\free{\phi}$. 
We define an $( \ell |U(\calD)|)$-vertex $\ell$-uniform hypergraph $\uhg(\phi,\calD)$ as follows.
\begin{itemize}
    \item $V(\uhg(\phi,\calD)) = \bigcup_{i=1}^\ell U_i(\calD)$
    \item $E(\uhg(\phi,\calD)) =\{\{(v_1,1),\dots,(v_\ell,\ell)\} \mid \exists \tau \in \ans{\phi, \calD}~ \forall i \in[\ell]: \tau(x_i)=v_i \}$
\end{itemize}
\end{defn}

\begin{obs}\label{obs:hyperGconstruction}
 Given an ECQ $\phi$ and a database $\calD$ with $\sig(\phi) \subseteq \sig(\calD)$,  the hyperedges of   $\uhg(\phi, \calD)$ 
 are in bijection with the elements of
 $\ans{\phi, \calD}$.
\end{obs}

By Observation~\ref{obs:hyperGconstruction}, the problem of approximating  $|\ans{\phi,\calD}|$ reduces immediately to approximating the number of hyperedges of $\uhg(\phi,\calD)$.
The latter can be achieved by Theorem~\ref{thm:colreduction} as long as we can (efficiently) simulate the oracle to evaluate ${\EdgeFree}(\uhg(\phi,\calD)[V_1,\ldots,V_\ell])$ for any $\ell$-partite subset $(V_1,\ldots,V_\ell)$ of $V(\uhg(\phi,\calD))$. 

We will show later that the most important
case is where, for each $i\in [\ell]$,
  $V_i \subseteq U_i(\calD)$.
It turns out that, in this case,
the evaluation of the predicate
${\EdgeFree}(\uhg(\phi,\calD)[V_1,\ldots,V_\ell])$   can be reduced to deciding the existence of a homomorphism between two structures which, intuitively, can be viewed as coloured versions of $\calA(\phi)$ and $\calB(\phi, \calD)$. 
We define these coloured versions in Definitions~\ref{def:hatH} and~\ref{def:hatG} respectively.
The colouring arises in Definition~\ref{def:hatG} 
in the following manner.
Let~$r$ and~$b$ be two colours. To handle disequalities, 
we introduce 
a collection of colouring functions $\boldf = \{f_\eta\}$,
where for each  $\eta=\{x_i,x_j\} \in \Delta(\phi)$,
$f_\eta$ is a function $f_\eta: U(\calD) \mapsto \{r,b\}$.

\begin{defn}[$\widehat{\calA}(\phi)$]\label{def:hatH}
Let $\phi$ be an ECQ. Let $\ell = |\free{\phi}|$ and
$k = |\vars{\phi}|-\ell$. Let $\{x_1,\ldots,x_{\ell+k}\}$
be an enumeration of the variables in~$\vars{\phi}$.
Recall the definition of $\calA(\phi)$ from Definition~\ref{def:calH}.
The structure $\widehat{\calA}(\phi)$ is a modification of $\calA(\phi)$ defined as follows.
\begin{itemize}
    \item $U(\widehat{\calA}(\phi))=\vars{\phi} = U(\calA(\phi))$.
    \item For each $R\in \sig(\calA(\phi))$, $R^{\widehat{\calA}(\phi)} = R^{\calA(\phi)}$.
\item For each variable $x_i$ of~$\phi$,
 $\widehat{\calA}(\phi)$
 has a   unary relation $P_i^{\widehat{\calA}(\phi)} := \{x_i\}$.
\item For each $\eta = \{x_i,x_j\}$ in~$\Delta(\phi)$ with $i<j$,
$\widehat{\calA}(\phi)$ has unary relations $R_\eta^{\widehat{\calA}(\phi)}:=\{x_i\}$ and $B_\eta^{\widehat{\calA}(\phi)} :=\{x_j\}$ .
 \end{itemize}
\end{defn}

\begin{obs}\label{obs:sizeOfhatA}\label{obs:A3}
$\widehat{\calA}(\phi)$ is obtained from $\calA(\phi)$ by adding 
$|\vars{\phi}| + 2 |\Delta(\phi)|$ unary relations. 
Since there are at most $
\binom{
|\vars{\phi}|}{2}$ disequalities in $\Delta(\phi)$, 
at most $|\vars{\phi}|^2$
unary relations are added in all.
Thus, 
$\structsize{\widehat\calA(\phi)} \leq
\structsize{\calA(\phi)} + 
2   | \vars{\phi}|^2$.
By Observation~\ref{obs:A1},
$\structsize{\widehat\calA(\phi)}  \leq 
 3 \structsize{\phi}  
 + 
2   | \vars{\phi}|^2
\leq 5
\structsize{\phi}^2$.
\end{obs}

\begin{defn}[$\widehat{\calB}(\phi,\calD,V_1,\ldots,V_\ell,\boldf)$]\label{def:hatG}\label{def:BB}
Let $\phi$ be an ECQ and let $\calD$ be a database with $\sig(\phi) \subseteq \sig(\calD)$.  
Let $\ell = |\free{\phi}|$ and
$k = |\vars{\phi}|-\ell$. 
 Let $(V_1,\ldots,V_\ell)$ be
an $\ell$-partite subset of $V(\uhg(\phi,\calD))
= \bigcup_{i=1}^\ell U_i(\calD)$
(from Definition~\ref{def:H},
recall  that $U_i(\calD)= U(\calD) \times \{i\}$ from 
Definition~\ref{def:Ui}).
Let $\{x_1,\ldots,x_\ell\}$ be an enumeration of the variables
in~$\free{\phi}$ and
let $\{x_1,\ldots,x_{\ell+k}\}$
be an enumeration of the variables in~$\vars{\phi}$.
Recall the definition of $\calB(\phi,\calD)$ from Definition~\ref{def:calG}.
The structure $\widehat{\calB}(\phi,\calD,V_1,\ldots,V_\ell,
\boldf)$ is a modification of $\calB(\phi,\calD)$ defined as follows.
To avoid notational clutter, we will just write $\widehat{\calB}$, rather than $\widehat{\calB}(\phi,\calB,V_1,\ldots,V_\ell,\boldf)$.
\begin{itemize}
\item For each $i\in\{1,\ldots,\ell\}$ let $S_i = V_i $ and for each $i\in \{\ell+1,\ldots,\ell+k\}$ let $S_i =  U_i(\calD)$. Define the universe of~$\widehat\calB$ as
$U(\widehat{\calB})=\bigcup_{i=1}^{\ell+k} S_i$.

\item 
  For each arity-$a$ relation symbol $R\in \sig(\calB(\phi,
 \calD))$,
  $\widehat{\calB}$ has the arity-$a$
  relation 
  $$R^{\widehat{\calB}}  := \{((w_1,i_1),\dots,(w_a,i_a)) \in U(\widehat{\calB})^a \mid (w_1,\dots,w_a)\in R^{\calB(\phi,
 \calD)} \}.$$ 
 Here $i_1,\ldots,i_a$ are any values such that
 $((w_1,i_1),\ldots,(w_a,i_a)) \in U(\widehat{\calB})^a$.
\item For each variable $x_i$ of~$\phi$, 
$\widehat{\calB}$
has a unary  relation  $P_i^{\widehat{\calB}} := S_i$. 
\item For each $\eta \in \Delta(\phi)$, we add 
to $\widehat{\calB}$ the unary relation 
 $ R_\eta^{\widehat{\calB}} := \{(x_i, j) \in U(\widehat{\calB}) \mid f_\eta(x_i)=r\}$
 and the unary relation
$ 
 B_\eta^{\widehat{\calB}} := \{(x_i, j) \in U(\widehat{\calB}) \mid f_\eta(x_i)=b\}$.  Here, $j$ is any value such that 
 $(x_i, j) \in U(\widehat{\calB}$.
\end{itemize}
\end{defn}

\begin{obs}\label{obs:sizeOfhatB}
To avoid notational clutter, let $\widehat{\calB}=\widehat{\calB}(\phi,\calD,V_1,\ldots,V_\ell,\boldf)$.
 Note that $|U(\widehat{\calB})|\leq
 \vars{\phi} \cdot|U(\calD)|
 =
 |\vars{\phi}|\cdot|U(\calB(\phi, \calD))|$.
 Let $a = \arity(\sig(\phi)) = \arity(\sig(\calB(\phi,\calD)))$.
 Assume that $\phi$ is not trivial, so $a\geq 1$.
 For  each relation $R^{\calB(\phi, \calD)}$ of $\calB(\phi, \calD)$,
 $|R^{\widehat{\calB}}|\leq |\vars{\phi}|^a\cdot|R^{\calB(\phi, \calD)}|  $. Additionally, we add  
 at most $|\vars{\phi}|^2$
unary relations  to $\widehat{\calB}$, each of size at most 
$|U(\widehat{\calB})|\leq
 |\vars{\phi}|\cdot|U(\calD)|$. Therefore,
\begin{align*}
\structsize{\widehat{\calB} } 
&\leq
|\sig(\widehat{\calB})| + |U(\widehat{\calB})| + 
|\vars{\phi}|^a \sum_{R \in \sig(\calB(\phi,\calD))} |R^{\calB(\phi,\calD)}| \arity(R)
+ |\vars{\phi}|^2
|\vars{\phi}|\cdot|U(\calD)|
\\
&\leq
 {| \vars{\phi}|}^{a}\cdot \structsize{\calB(\phi, \calD)}
+ {| \vars{\phi}|}^2
+ {|\vars{\phi}|}^3  |U(\calD)|.
\end{align*}
If $\phi$ has $\nu$ negated predicates then,
Observation~\ref{obs:A2} guarantees that
$$\structsize{\calB(\phi,\calD)}\le  
2\structsize{\phi}(
 \structsize{\calD} + \nu   {|U(\calD)|}^a)
$$ so,  plugging this in,
$\structsize{\widehat{\calB}(\phi,\calD,V_1,\ldots,V_\ell,\boldf)}
 \leq
    \exp(\bigO({\structsize{\phi}}^2)) \cdot
    (    \structsize{\calD} +
    { \nu |U(D)|^{a}}  
    )    $.

\end{obs}

Note that $\widehat{\calA}(\phi)$ and 
 $\widehat{\calB}(\phi,\calD,V_1,\ldots,V_\ell,\boldf)$
 have the same signature by construction.
Recall that our main goal is to simulate the oracle for
${\EdgeFree}(\uhg(\phi,\calD)[V_1,\ldots,V_\ell])$,
and that we have stated (but not yet proved) that the most important case is when each $V_i \subseteq U_i(\calD)$.
The following lemma establishes certain properties of
the structures~$\widehat{\calA}(\phi)$ 
and $\widehat{\calB}(\phi,\calD,V_1,\ldots,V_\ell,\boldf)$ that 
apply in this case.

 \newcommand{\statelemmaoracleconstruction} 
 {\sl Let $\phi$ be an ECQ and let $\calD$ be a database with 
$\sig(\phi) \subseteq \sig(\calD)$. 
Let $\ell = |\free{\phi}|$. 
 Let $(V_1,\ldots,V_\ell)$ be
an $\ell$-partite subset of $V(\uhg(\phi,\calD))$.
Suppose that for each $i\in [\ell]$,
  $V_i\subseteq U_i(\calD)$.
  Then the 
hypergraph
 $ \uhg(\phi,\calD)[V_1,\ldots,V_\ell]$ has a hyperedge
 if and only if there is 
 a collection~$\boldf$ of colouring functions such that there is
 a homomorphism from
$\widehat{\calA}(\phi)$
to 
$\widehat{\calB}(\phi,\calD,V_1,\ldots,V_\ell,\boldf)$.}

\begin{lem} \label{lem:oracleconstruction}
\statelemmaoracleconstruction
\end{lem}
\begin{proof} 
 
 Let  
$k = |\vars{\phi}|-\ell$.  
Let $\{x_1,\ldots,x_\ell\}$ be an enumeration of the variables
in~$\free{\phi}$ and
let $\{x_1,\ldots,x_{\ell+k}\}$
be an enumeration of the variables in~$\vars{\phi}$.
 We consider both directions.
 
 \begin{itemize}
     \item 
 If $\uhg(\phi,\calD)[V_1,\ldots,V_\ell]$  
 has any hyperedges then, by Definition~\ref{def:l-uniformHypergraph}, it must contain a hyperedge of the form $\{(v_1,1),\dots,(v_\ell,\ell)\}$, where $(v_i,i)\in V_i$ for all~$i\in[\ell]$,
 and there is an assignment
 $ \tau \in \ans{\phi, \calD}$ such that  $\forall i \in[\ell]$ we have $\tau(x_i)=v_i$. 
 Consequently 
 by~\eqref{equ:Ans}
 there is a homomorphism~$h$ from~$\calA(\phi)$ to~$\calB(\phi, \calD)$ that extends~$\tau$ and satisfies all disequalities in~$\Delta(\phi)$, that is, $h(x_i)\neq h(x_j)$ for all $ \{x_i,x_j\}\in \Delta(\phi)$.
 
 Now fix such a homomorphism~$h$ and choose 
 the collection~$\boldf$ of colouring functions
 so that for each $\eta=\{x_i,x_j\} \in \Delta(\phi)$ with  $i<j$, 
 $f_\eta$ maps  $h(x_i)$ to~$r$ and $h(x_j)$ to~$b$, which is possible since $h(x_i)\neq h(x_j)$.
 The values of $f_\eta$ on variables other than~$x_i$ and~$x_j$ are irrelevant and can be chosen arbitrarily.
  
 Using $h$, we construct a homomorphism~$\hat{h}$
from $\widehat{\calA}(\phi)$ to $\widehat{\calB}=\widehat{\calB}(\phi,\calD,V_1,\dots,V_\ell,\boldf)$, proceeding as follows.  
The elements of $U(\widehat{\calA}(\phi))$ are the variables $x_i$ of~$\phi$.
For each $x_i$, take $\hat{h}(x_i) = (h(x_i),i)$.
The fact that the non-unary 
relations are preserved by $\hat{h}$ follows from the fact that~$h$ is a homomorphism and from the definition of the predicates   $R^{\widehat{B}}$. 
To preserve $P_i^{\widehat{\calA}(\phi)}$, $\hat{h}$ must map 
$x_i$ to an element in~$S_i$. In particular, for each free
variable~$x_i$, we have $\hat{h}(x_i)=(h(x_i),i) = (v_i,i) \in V_i$, which was noted above.

Now consider
 $\eta = \{x_i,x_j\}\in \Delta(\phi)$  with $i<j$.
 To preserve 
 $R_\eta^{\widehat{\calA}(\phi)}$ and
  $B_\eta^{\widehat{\calA}(\phi)}$, it must be the case that $f_\eta(h(x_i)) = r$
  and $f_\eta(h(x_j))=b$, which, however, is satisfied for our choice of $f$. Thus, $\hat{h}$ is a homomorphism as desired. 

\item  
In the other direction, suppose for some 
collection~$\boldf$ of colouring functions
   that there is a homomorphism $\hat{h}$
from
 $\widehat{\calA}(\phi)$ to $\widehat{\calB}=\widehat{\calB}(\phi,\calD,V_1,\dots,V_\ell,\boldf)$. 
 The relation~$P_i$ ensures  that $\hat{h}$ maps each $x_i$ to~$S_i$.
 So for each $i\in[\ell]$, $\hat{h}(x_i) = (w_i,i)\in V_i$ for some~$w_i\in U(\calD)$. 
 
 We construct a homomorphism~$h$ 
 from~$\calA(\phi)$ to~$\calB(\phi,\calD)$ 
 by setting $h(x_i) = w_i$ for all variables $x_i\in \vars{\phi}$.
 The 
 relations from the signature of~$\calA(\phi)$   ensure that~$h$ is a homomorphism from~$\calA(\phi)$ to~$\calB(\phi,\calD)$. 
 Let $\tau = \proj(h,\free{\phi})$.
 We will show that 
 $ \tau \in   \ans{\phi, \calD} $.
 By Definition~\ref{def:l-uniformHypergraph}, that implies that
   $\{(\tau(x_1),1),\dots,(\tau(x_\ell),\ell)\}$ is a hyperedge of
  $\uhg(\phi,\calD)[V_1,\ldots,V_\ell]$, completing the proof.
 
 So it remains   to prove that 
  $ \tau \in   \ans{\phi, \calD} $.
  By~\eqref{equ:Ans}
  it suffices to show that  for all $\{x_i,x_j\} \in \Delta(\phi)$, we have $h(x_i)\neq h(x_j)$.
  To see this, 
  consider $\{x_i,x_j\} \in \Delta(\phi)$ and
  suppose that $i<j$.
  Then the relation $R_\eta$ ensures
  that $h(x_i)$ is in   $R_\eta^{\widehat{\calB}} $  
  and that $h(x_j)$ is in 
  $B_\eta^{\widehat{\calB}} $.
  Since these two relations are disjoint,  we find that $h(x_i)\neq h(x_j)$, as required.

\end{itemize}

\end{proof}

Given   Lemma~\ref{lem:oracleconstruction},
we will be able to
use colour-coding to 
simulate the  oracle  for
${\EdgeFree}(\uhg(\phi,\calD)[V_1,\ldots,V_\ell])$ using
an oracle for the decision homomorphism problem. 
Colour-coding is common in parameterised algorithms
and our application is similar to the approach that has
been used in the decision setting   by Papadimitriou and Yannakakis~\cite{Papadimitriou99:JCSS} and Koutris et al.~\cite{Koutris17:TCS}. 
Using this, we
are now able to prove Lemma~\ref{lem:interface}, which we restate here for convenience.

\begin{leminterface}
    \stateleminterface
\end{leminterface}

\begin{proof} 
Let $(\phi,\calD,\epsilon,\delta)$ be an input to the algorithm.
Let $\ell = |\free{\phi}|$ and
$k = |\vars{\phi}|-\ell$. 
Let $\{x_1,\ldots,x_\ell\}$ be an enumeration of the variables
in~$\free{\phi}$ and
let $\{x_1,\ldots,x_{\ell+k}\}$
be an enumeration of the variables in~$\vars{\phi}$.

To avoid notational clutter, we set $\calA:=\calA(\phi)$, $\calB:=\calB(\phi,\calD)$, $\Delta:=\Delta(\phi)$, 
$H:= \uhg(\phi,\calD)$ and $V := V(H)$.
Let $N = |V|$. Note that, by construction (definition~\ref{def:H}), $N = \ell |U(\calD)|$.

The goal of the algorithm is to provide an 
 $(\epsilon,\delta)$-approximation of  $|\ans{\phi,\calD}|$.
 By  Observation~\ref{obs:hyperGconstruction},
 this is the same as providing an
 $(\epsilon,\delta)$-approximation of $|E(H)|$.

\paragraph*{Simulating the oracle calls}
Our goal is to apply Theorem~\ref{thm:colreduction}
with $\varepsilon, \delta/2$, and $H$.
To do this, we must provide a simulation strategy for an oracle query ${\EdgeFree}(H[W_1,\ldots,W_\ell])$, where $(W_1,\ldots,W_\ell)$ is any $\ell$-partite subset of $V$. 
We must ensure that the probability that any simulation of the
oracle fails (during the whole run of the
algorithm from Theorem~\ref{thm:colreduction}) is at most~$\delta/2$.
To do this,
we provide a simulation strategy 
for an individual oracle 
call with failure probability at most~$\delta/(2T)$, where 
$T = \Theta(\log(1/\delta)\varepsilon^{-2}\ell^{6\ell}(\log N)^{4\ell+7})$
is the upper bound on the number of calls to ${\EdgeFree}$ in Theorem~\ref{thm:colreduction}.
Note (by a union bound) that this implies that the overall probability that
any oracle call fails is at most $\delta/2$.

We simulate an arbitrary oracle call ${\EdgeFree}(H[W_1,\ldots,W_\ell])$ by evaluating $\ell!$ more restricted oracle calls, each 
with failure probability at most $\delta/(2 T \ell!)$ (which gives the desired failure probability of at most~$\delta/(2T)$ by a union bound).
Each of the restricted oracle calls is of the form
${\EdgeFree}(H[V_1,\ldots,V_\ell])$
where each  $V_i \subseteq U_i(\calD)$.
The restricted oracle calls are chosen by considering each
  permutation $\pi$ of $[\ell]$ and  setting 
  $V'_i=W_i\cap U_{\pi(i)}(\calD) $. 
  Since (from Definition~\ref{def:H})
  each hyperedge of~$H$ contains exactly one element from each 
  $U_i(\calD)$ with $i\in[\ell]$,
  it is the case that
  $H[W_1,\ldots,W_\ell]$ has a hyperedge if 
  and only if there is a permutation~$\pi$ so that
  $H[V'_1,\ldots,V'_\ell]$ has a hyperedge.
  Equivalently, setting $V_i = V'_{\pi(i)}$, 
  $H[W_1,\ldots,W_\ell]$ has a hyperedge if 
  and only if there is a permutation~$\pi$ so that
  $H[V_1,\ldots,V_\ell]$ has a hyperedge.

 By
 Lemma~\ref{lem:oracleconstruction},
 simulating the restricted oracle call 
 ${\EdgeFree}(H[V_1,\ldots,V_\ell])$
 is equivalent to checking 
 whether 
 there is 
 a collection~$\boldf$ of colouring functions such that there is
 a homomorphism from
$\widehat\calA :=\widehat{\calA}(\phi)$
to 
$\widehat{\calB}(\phi,\calD,V_1,\ldots,V_\ell,\boldf)$. For a given 
collection~$\boldf$ of colouring functions, we use access to the oracle for $\decoracle$ to determine whether such a homomorphism exists.

We now consider how to 
simulate the restricted oracle call 
${\EdgeFree}(H[V_1,\ldots,V_\ell])$.
Let $Q = \lceil \log(2 T \ell!/\delta) \rceil 4^{|\Delta|}$.
We   
make $Q$ repetitions, each time choosing the collection~$\boldf$ of
colouring functions uniformly at random.
During a given repetition,   
we choose the collection~$\boldf$ as follows. For each $\eta \in \Delta$ and 
each $u\in U(\calD)$, with probability $1/2$, we set $f_\eta(u)=r$. Otherwise, we set $f_\eta(u)=b$. 
Having chosen~$\boldf$,  we query the oracle for $\decoracle$ to determine whether there is a homomorphism from $ \widehat{\calA}$
to 
$  \widehat{\calB}(\phi,\calD,V_1,\ldots,V_\ell,\boldf)$.
If, for any of the $Q$ choices of~$\boldf$, the oracle for $\decoracle$ 
finds a homomorphism, the simulation announces that 
$H[V_1,\ldots,V_\ell]$ has a hyperedge. Otherwise, 
it announces that $H[V_1,\ldots,V_\ell]$ has no hyperedge.

\paragraph*{Failure probability of the simulation}
We now bound the failure probability of the simulation.
If there is no~$\boldf$ such that there is a 
 homomorphism from $ \widehat{\calA}$
to 
$  \widehat{\calB}(\phi,\calD,V_1,\ldots,V_\ell,\boldf)$ then the simulation is correct.
Suppose instead that there is a collection~$\boldf'$ of colouring
functions such that there is a homomorphism~$h$ 
from  $ \widehat{\calA}$
to 
$  \widehat{\calB}(\phi,\calD,V_1,\ldots,V_\ell,\boldf')$.
When~$\boldf$ is chosen uniformly at random during one of the
repetitions, the probability that 
$h$ is a homomorphism from 
 $ \widehat{\calA}$
to 
$  \widehat{\calB}(\phi,\calD,V_1,\ldots,V_\ell,\boldf)$
is at least the probability that, for each 
$\eta=\{x_i,x_j\}\in \Delta$, we have
  $f_\eta(x_i)=
f'_\eta(x_i)$ and 
$f_\eta(x_j)=f'_\eta(x_j)$.
This probability is at least $4^{-|\Delta|}$.
Thus, the probability that
all $Q$ guesses fail is at most 
$(1-4^{-|\Delta|})^Q  
\leq \exp(-Q4^{-|\Delta|}) \leq \delta/(2 T \ell!)$, as required.

\paragraph*{Size of the constructed structures}
For a fixed $\ell$-partite subset $(V_1,\ldots,V_\ell)$ of~$V$ and collection~$\boldf$ of colouring functions, 
consider the structures
 $\widehat\calA$
and
$\widehat{\calB}(\phi,\calD,V_1,\ldots,V_\ell,\boldf)$.
It follows from
Observation~\ref{obs:A3} that  $\widehat{\calA}$ is obtained from $\calA$ by adding unary relations  
and that
$\structsize{\widehat\calA}  \leq 
   5
\structsize{\phi}^2$ (as required).
It is clear from Definitions~\ref{def:calH} and~\ref{def:hatH} that the time needed to construct~$\widehat{\calA}$ is linear in its size.
 Moreover, Observation~\ref{obs:sizeOfhatB}
 guarantees that
 $\structsize{\widehat{\calB}(\phi,\calD,V_1,\ldots,V_\ell,\boldf)}
 \leq
    \exp(\bigO({\structsize{\phi}}^2)) \cdot
    (    \structsize{\calD} +
    { \nu |U(D)|^{a}}  
    )  $. 
    It is clear from
 Definition~\ref{def:calG} and~\ref{def:hatG} 
 that the time needed to construct~$\widehat{\calB}$ is linear in its size.

\paragraph*{Runtime analysis}
First, we bound the number $X \leq T \ell! Q$ of oracle calls.
Plugging in the definition of~$Q$ and applying crude upper bounds, we have
$X = \bigO(T^2 (\ell!)^2 \log(1/\delta)     4^{|\Delta|} )$.
Plugging in the definition of~$T$ and pulling out $\log(1/\delta)$ 
and $\epsilon^{-1}$~factors,
we have
$$X = \poly(\log(1/\delta),\epsilon^{-1}) 
\bigO(  \ell^{12\ell} (\log N)^{8\ell+14} (\ell!)^2      4^{|\Delta|} )\,.$$
 Since $\ell^\ell$, $\ell!$  and $4^{|\Delta|}$
 are all at most $\exp(\bigO(\structsize{\varphi}^2))$
 we obtain 
$$X = \poly(\log(1/\delta),\epsilon^{-1}) \cdot
\exp(\bigO(\structsize{\varphi}^2))\cdot
\bigO(   (\log N)^{8\ell+14}       ).$$
If $\log N < e^{\ell} $ 
then final term can be subsumed by the $\exp(\bigO(\structsize{\varphi}^2))$ term.
Otherwise, the final term is 
$\bigO(N)= \bigO( \ell |U(\calD)|)$. To see this, note that
\[ (\log N)^{8\ell+14} = (e^{\ell})^{8 \log\log N}\cdot (\log N)^{14}\leq (\log N)^{8 \log \log N + 14} \in O(N)\]
Therefore,
 $$X = \poly(\log(1/\delta),\epsilon^{-1}) \cdot
\exp(\bigO(\structsize{\varphi}^2))\cdot
\bigO(  |U(\calD)|      ).$$

Second, consider the time needed to construct an oracle query 
$\decoracle(\widehat\calA,{\widehat\calB})$.
It is dominated by the time needed
to construct $\widehat{\calB}$, and the total running time is bounded from above by
 $$\poly(\log(1/\delta),\epsilon^{-1}) \cdot
\exp(\bigO(\structsize{\varphi}^2))\cdot
\bigO(  |U(\calD)|      )\cdot
 \exp(\bigO({\structsize{\phi}}^2)) \cdot
    (    \structsize{\calD} +
    { \nu |U(D)|^{a}}  
    ) ,$$
 which can easily be simplified to prove the lemma.\end{proof}

\section{FPTRAS for \#ECQ with bounded Treewidth and Arity}\label{sec:bd}

The goal of this section is to establish Theorem~\ref{thm:mainbounded}, i.e., to construct an FPTRAS for \#ECQ on classes of queries whose hypergraphs have bounded treewidth and arity. Thanks to Lemma~\ref{lem:interface}, it will be sufficient to rely on an efficient algorithm for the decision version of the homomorphism problem. For the case of bounded arity, we will use the algorithm due to Dalmau et al.\ \cite{Dalmau02:CP} (see \cite[Theorem~3.1]{Grohe07:jacm} for an explicit statement of the extension from graphs to structures). 

For each structure $\calA$, there is an associated \emph{hypergraph of} $\calA$, which we denote $H(\calA)$. Its vertices are $V(H(\calA))=U(\calA)$, and $H(\calA)$ has a hyperedge $e$ whenever there is a relation $R\in\sig(\calA)$ with $e\in R^\calA$. 
The treewidth of $\calA$ is the treewidth of its associated hypergraph.\footnote{Note that the hypergraph $H(\varphi)$ of the query~$\phi$ and the hypergraph of its associated structure $H(\calA(\varphi))$ are the same. Consequently, the treewidth of $\varphi$ and $\calA(\varphi)$ are equal.}

Given a class of structures $\calS$, we write $\decoracle(\calS)$ for the restriction of $\decoracle$ in which the input $(\calA,\calB)$ must satisfy $\calA\in\calS$.

\begin{thm}[\cite{Dalmau02:CP,Grohe07:jacm}]\label{thm:Wednesday}
Let $t$ and $a$ be positive integers. Let $\calS$ be a 
class of structures such that every structure in $\calS$ has treewidth at most $t$ and arity at most $a$. Then $\decoracle(\calS)$ is polynomial-time solvable.
\end{thm}

We remark that  Theorem~\ref{thm:Wednesday} is in fact a weaker version of the corresponding theorem in~\cite{Dalmau02:CP,Grohe07:jacm}, which also applies to classes $\calS$ in which only the \emph{homomorphic cores} of members of $\calS$ have bounded treewidth. We do not need the more general version here.
The proof of Theorem~\ref{thm:Wednesday} relies on the fact
that there is a polynomial-time algorithm that takes as input a structure $\calA \in \calS$ and produces a tree decomposition 
of $H(\calA)$ with treewidth at most~$4t+4$
(the exact treewidth of the tree decomposition that is produced is not important, but it is important that it is $\bigO(t)$).

We are now able to prove Theorem~\ref{thm:mainbounded}, which we restate here for convenience.

\begin{thmmainbounded}
   \statethmmainbounded
\end{thmmainbounded}

\begin{proof} 
  Let $t$, $a$ and $C$ be as in the statement. Let $\calS$ be the class of all structures with treewidth at most~$t$ and arity at most~$a$. 

  By  Theorem~\ref{thm:Wednesday}, there is a polynomial-time 
  algorithm~\homalg\ for $\decoracle(\calS)$.
  
  The desired FPTRAS is now obtained by using Lemma~\ref{lem:interface} and simulating the oracle using the algorithm~\homalg.
  Given an input $(\phi,\calD)$ with $\phi \in \Phi_C$
  and $\sig(\phi) \subseteq \sig(\calD)$,
  we wish to approximate 
  $|\ans{\phi,\calD}|$.
  By the statement of Lemma~\ref{lem:interface},
  every oracle query $(\widehat{\calA},\widehat{\calB})$ 
  produced in the course of the approximation
  has the property that $\widehat{\calA}$ is obtained from $\calA(\varphi)$ by adding unary relations. 
  Our goal is to show how to simulate
  the oracle call $\decoracle(\widehat{\calA},\widehat{\calB})$.
  
  In order to use Algorithm~\homalg{} for the simulation,
  we need only show that $\widehat\calA$ has 
  treewidth at most~$t$ and arity at most~$a$.
  Equivalently, we need to show that $H(\widehat\calA)$ has
  treewidth at most~$t$ and arity at most~$a$.
  
  Since $H(\phi)\in C$, we know that $H(\calA(\phi))$ has
  treewidth at most~$t$ and arity at most~$a$.
  Since $H(\widehat\calA)$ is obtained from $H(\calA(\phi))$
  by adding size-$1$ hyperedges and $a\geq 1$, it is clear that the arity of
  $H(\widehat\calA)$ is at most~$a$.
  To see that the treewidth of 
  $H(\widehat\calA)$ is at most~$t$, consider
  any tree decomposition $(T,\boldB)$ of~$H(\calA(\phi))$ 
  with $\tw{T,\boldB}=t$.
  Construct a tree decomposition $(T',\boldB')$ of
  $H(\widehat\calA)$ as follows.
  Start by setting $(T',\boldB')=(T,\boldB)$. Then,
  for every $v\in \vars{\phi}$, check whether
  \begin{enumerate}
      \item $\{v\}$ is a hyperedge in  $E(H(\widehat\calA)) $, and
      \item $v$ is not in any bag $B_t$ of $T$
  \end{enumerate}
  If this occurs, add a new leaf $t'$ to $T'$  with $B_{t'} = \{v\}$.
  It is clear that $(T',\boldB')$ is a tree decomposition of 
  $H(\widehat\calA)$ with treewidth~$t$.

  Now by Lemma~\ref{lem:interface} and Theorem~\ref{thm:Wednesday}, 
  the total running time is bounded by
  \[\exp(\bigO(||\varphi||^2)) \cdot \poly(\log(1/\delta),\varepsilon^{-1},||\calD||) \cdot \poly(||\widehat{\calA}||+||\widehat{\calB}||) \,,\]
  where $(\widehat{\calA},\widehat{\calB})$ is the oracle query that maximises $||\widehat{\calA}||+||\widehat{\calB}||$.
  
  Since the size $||\widehat{\calA}||+||\widehat{\calB}||$ must be bounded by $ \exp(\bigO(||\varphi||^2)) \cdot \poly(\log(1/\delta),\varepsilon^{-1},||\calD||)$
  (by the bound on the running time in Lemma~\ref{lem:interface}), and $\poly(\exp(\bigO(||\varphi||^2)))=\exp(\bigO(||\varphi||^2))$ we can bound the overall running time by
  \[\exp(\bigO(||\varphi||^2)) \cdot \poly(\log(1/\delta),\varepsilon^{-1},||\calD||) \,,\]
  which concludes the proof.
\end{proof}

\section{Beyond Bounded Arity}\label{sec:beyond}

Theorem~\ref{thm:mainbounded} gives an FPTRAS for $\ECQ{\Phi_C}$ 
where $\Phi_C$ is a set of ECQs~$\phi$ for which $H(\phi)$ has bounded treewidth and bounded arity.  There are many notions of width that refine treewidth, but in the bounded arity case, it is also known that bounding these widths leads to the same class of queries as bounding treewidth.

When the arity restriction is relaxed, the different notions of width play a bigger role.
In this section we examine the difficulty of $\DCQ{\Phi_C}$ and 
$\CQ{\Phi_C}$  when $\Phi_C$ does not impose a bound on arity.
In Section~\ref{sec:sw} we extend Theorem~\ref{thm:mainbounded} to this setting (in the case of $\DCQ{\Phi_C}$), obtaining a stronger result in terms of   adaptive width. In Section~\ref{sec:fh} we improve the result of~\cite{ArenasNew} 
which gives an FPRAS for~$\CQ{\Phi_C}$
by bounding fractional hypertreewidth rather than hypertreewidth.

\subsection{FPTRAS for \#DCQ with bounded Adaptive Width}\label{sec:sw}

The goal of this section is to prove Theorem~\ref{thm:submod}.
We start by defining the hypergraph width
measures that we need, following the framework of~\cite{Adler06:phd}.

\begin{defn}[$f$-width] \label{def:fwidth}
Let $H$ be a hypergraph. For any function $f:2^{V(H)}\to {\mathbb R}_{\geq 0}$, the
\emph{$f$-width}
of a tree decomposition $(T, \boldB)$ of $H$ is the
maximum of $f(B_t)$ taken over all $t \in V(T)$. The \emph{$f$-width of $H$},
denoted by $f(H)$, is the minimum $f$-width over all tree decompositions of $H$. 
\end{defn}

It is clear from Definition~\ref{def:tw} that
 for $f(X) = |X|-1$, the $f$-width of $H$ is identical to the treewidth of $H$.
 
 \begin{defn}[fractional independent set, adaptive width] 
A \emph{fractional independent set} of a hypergraph $H$ is a function
$\mu:V(H)\to[0,1]$ such that for all $e\in E(H)$, we have
$\sum_{v\in e}\mu(v)\leq 1$. 
For $X\subseteq V(H)$, we define
$\mu(X)=\sum_{v\in X}\mu(v)$.
The \emph{adaptive width}~\cite{Marx11:tocs-tables} of
$H$, denoted by $\aw{H}$, is the supremum of $\mu$-width$(H)$, where
the supremum is taken over all fractional independent sets $\mu$ of $H$.
\end{defn}
 
 Recall from Lemma~\ref{lem:dom} that treewidth  is strongly dominated by adaptive width. For the sake of completeness, 
 we observe that this strict domination  requires unbounded arity.
 \begin{obs}\label{obs:tw_vs_aw}
 Let $H$ be a hypergraph with arity $a$. Then $\tw{H}\leq a \cdot \aw{H}-1$.
 \end{obs}
 \begin{proof}
   If $a=0$ then $H$  has no hyperedges and thus $\tw{H}= -1 = a \cdot \aw{H}-1$. Otherwise, set $\mu(v):= 1/a$ and observe that $\mu$ is a fractional independent set. Hence there exists a tree decomposition $(T, \boldB)$ of $\mu$-width$(H)$ at most $\aw{H}$, that is, for each $t\in V(T)$ we have $\mu(B_t) \leq \aw{H}$. Since $\mu(B_t)= |B_t|/a$, we conclude that $|B_t| \leq a \cdot \aw{H}$ and thus the treewidth is bounded by $a \cdot \aw{H}-1$. 
 \end{proof}
 
 Theorem~\ref{thm:submod} improves Theorem~\ref{thm:mainbounded} for the special case 
where the queries are DCQs.
In particular, it gives an FPTRAS for 
$\DCQ{\Phi_C}$
 for every class~$C$ of hypergraphs with bounded adaptive width.
 The FPTRAS uses Lemma~\ref{lem:interface}.
 In order to use this lemma, we first 
 show that adding 
 unary relations to a structure cannot increase its adaptive width
 in any harmful way.

\begin{lem}\label{lem:unarysubmod}
Let $\calA$ be a structure.  
Let $\widehat\calA$ be a structure with universe $U(\widehat\calA)=U(\calA)$ and signature $\sig(\widehat\calA) = \sig(\calA)\cup \rho$ 
where $\rho$ is a set of arity-$1$ relation symbols.
Suppose that,  for each $R\in \sig(\calA)$, we have that $R^{\widehat\calA}= R^{\calA}$.
Then $\aw{\widehat\calA} \leq \max\{\aw{\calA},1\}$.
\end{lem}
\begin{proof}
Let $H=H(\calA)$, and let $\widehat{H}=H(\widehat{\calA})$. Observe that $\widehat{H}$ is obtained from $H$ by adding edges of arity $1$.  Also,
$V(\widehat{H})=V(H)$.
From now on, we will refer to this vertex set as~$V$.

We start by observing that every fractional independent set of~$H$ is a fractional independent set of~$\widehat{H}$, and vice-versa.
 
Let $\hat{b} = \max\{\aw{H},1\}$.
We claim that $\aw{\widehat{H}} \leq \hat{b}$.
To see this, consider any fractional independent set~$\mu$ of~$\widehat{H}$. 
We wish to show that $\mu$-width$(\widehat{H})\leq \hat{b}$,
which means that there is a tree decomposition 
$(T',\boldB')$ of~$\widehat{H}$ such that every $t\in V(T')$
satisfies $\mu(B_t)\leq \hat{b}$.

We start with a tree decomposition $(T,\boldB)$
such that the $\mu$-width of~$(T,\boldB)$ is at most~$\aw{H}$.
Consider the tree decomposition $(T',\boldB')$ of~$\widehat{H}$
constructed as in the proof of Theorem~\ref{thm:mainbounded}.
Start by setting $(T',\boldB')=(T,\boldB)$. Then,
  for every $v\in \vars{\phi}$, check whether
  \begin{enumerate}
      \item $\{v\}$ is a hyperedge in  $E(\widehat{H}) $, and
      \item $v$ is not in any bag $B_t$ of $T$
  \end{enumerate}
  If this occurs, add a new leaf $t'$ to $T'$  with $B_{t'} = \{v\}$.
The claim follows since,
for any $t\in V(T)$, 
$\mu(B_t) \leq \aw{H} \leq \hat{b}$.
Also, for any $t \in V(T') \setminus V(T)$,
$\mu(B_t) \leq 1 \leq \hat{b}$. 
We conclude that 
$\aw{\widehat{H}} \leq \hat{b}$, as desired.\end{proof}

Recall that, given a class of structures~$\calS$,
$\decoracle(\calS)$ is the problem of deciding,
given a pair $(\calA,\calB)$ of structures with $\calA \in \calS$,
whether there is a homomorphism from~$\calA$ to~$\calB$.
The following algorithmic result is due to Marx,
though we rephrase it in terms of homomorphisms between structures.
 
\begin{thm}[{\cite[Theorem 4.1]{Marx13:jacm}}]\label{thm:marxsubmodalgo}
Let $\hat{b}$ be a positive integer. Let $\calS$ be a class of structures such that every structure in~$\calS$ has adaptive width at most $\hat{b}$. Then  there is a 
(computable) 
function~$f$ such that
$\decoracle(\calS)$ can be solved in time
$f(\structsize{\calA})
  \cdot \poly(\structsize{\calA}+\structsize{\calB}) $.
\end{thm}

We are now able to prove Theorem~\ref{thm:DCQsubmod}, which we restate for convenience.

\begin{thmDCQsubmod}
    \stateDCQsubmod
\end{thmDCQsubmod}
\begin{proof}
Let $b$ and $C$ be as in the statement. Let $\calS$ be the class of all structures with adaptive width at most $\hat{b} = \max\{b,1\}$. By  Theorem~\ref{thm:marxsubmodalgo}, there is a (computable)  function~$f$, a polynomial~$p$ and an algorithm $\homalg$ 
for $\decoracle(\calS)$  
which, given input $(\widehat{\calA},\widehat{\calB})$,
runs in time  
$f(\structsize{\calA})
  \cdot \poly(\structsize{\calA}+\structsize{\calB}) $.

The desired FPTRAS is now obtained by using Lemma~\ref{lem:interface} and simulating the oracle using the algorithm~$\homalg$. 
Given an input $(\phi,\calD)$ with $\phi \in \Phi_C$
and $\sig(\phi) \subseteq \sig(\calD)  $,
we wish to approximate 
$|\ans{\phi,\calD}|$.
By the statement of Lemma~\ref{lem:interface},
every oracle query $(\widehat{\calA},\widehat{\calB})$ 
produced in the course of the approximation
has the property that  
    $\widehat{\calA}$ can be obtained from $\calA(\phi)$ by adding   unary relations  
and satisfies
$\structsize{\widehat\calA}  \leq 5\structsize{\phi}^2$.
Our goal is to show how to simulate
the oracle call $\decoracle(\widehat{\calA},\widehat{\calB})$.

In order to use algorithm~$\homalg$ for the simulation,
we need only show that $\widehat{\calA}$ has adaptive width at
most~$\hat{b}$. This follows from Lemma~\ref{lem:unarysubmod}.

Now by Lemma~\ref{lem:interface} and Theorem~\ref{thm:marxsubmodalgo}, 
since the input formula~$\phi$ has no negated predicates, that is, the quantity~$\nu$ in the statement
of Lemma~\ref{lem:interface} is~$0$, the total running time is  at most
\[\exp(\bigO(||\varphi||^2)) \cdot \poly(\log(1/\delta),\varepsilon^{-1},||\calD||) \cdot f(||\widehat{\calA}||)\cdot p(||\widehat{\calA}||+||\widehat{\calB}||) \,,\]
where $(\widehat{\calA},\widehat{\calB})$ is the oracle query that maximises $f(||\widehat{\calA}||)\cdot p(||\widehat{\calA}||+||\widehat{\calB}||)$.
  
Since~$\structsize{\widehat{\calB}}$ must be bounded by $\exp(\bigO(||\varphi||^2)) \cdot \poly(\log(1/\delta),\varepsilon^{-1},||\calD||)$
(by the running time in Lemma~\ref{lem:interface}), and $\structsize{\widehat{\calA}}$ is bounded by  $5\structsize{\phi}^2$, as we have already mentioned,
there is a (computable) function $\hat{f}$ such that the running time is bounded by
  
\[\hat{f}(||\varphi||) \cdot \poly(\log(1/\delta),\varepsilon^{-1},||\calD||)  \,,\]
 which yields the desired FPTRAS and thus concludes the proof.
\end{proof}

\subsection{FPRAS for \#CQ with bounded Fractional Hypertreewidth}\label{sec:fh}

Recall the definition of  
tree decomposition (Definition~\ref{def:tw})   ---
a \emph{tree decomposition} of a hypergraph $H$ is a pair 
$(T,\boldB)$ where $T$ is a (rooted) tree 
and $\boldB$ assigns a subset $B_t \subseteq V(H)$ 
(called a \emph{bag})
to each $t\in V(T)$. The following two conditions are satisfied: (i) for each $e \in
E(H)$ there exists $t \in V(T)$ such that $e \subseteq B_t$, and (ii) for each
$v \in V(H)$ the set $\{ t \in V(T) \mid v \in B_t\}$ induces a (connected)
subtree of $T$.

 We will use the following notation associated with a 
 tree decomposition. Let $t^*$ be the root of~$T$.
Given a   vertex $t\in V(T)$, let $T_t$ denote
the subtree of~$T$ rooted at~$t$.

 \begin{defn}[hypertree decomposition, guard, hypertreewidth]
A \emph{hypertree decomposition}~\cite[Definition A.1]{Gottlob02:jcss-hypertree} of a hypergraph~$H$ 
is a triple~$(T,B,\Gamma)$ where $(T,\boldB)$ is a tree decomposition of~$H$ 
and $\Gamma$ assigns a subset $\Gamma_t \subseteq E(H)$ (called a \emph{guard})
to each
$t\in V(T)$. In addition to the two conditions that $(T,\boldB)$ satisfies, 
the following conditions are satisfied:
(iii) for each   $t\in V(T)$, $B_t \subseteq \cup_{e \in \Gamma_t} e$.
  (iv)
for each $t\in V(T)$, 
$(\cup_{e \in \Gamma_t} e ) \cap 
(\cup_{t'\in V(T_t)} B_{t'}) \subseteq B_t$.
The hypertreewidth of the decomposition $(T,B,\Gamma)$ is the maximum cardinality of a guard.
The hypertreewidth of~$H$, denoted by $\hw{H}$,
is the minimum hypertreewidth of any hypergraph decomposition of~$H$.
\end{defn}

Arenas et al.~\cite[Theorem 3.2]{ArenasNew} prove the following result.
\begin{thm}[Arenas, Croquevielle, Jayaram, Riveros] \label{thm:Arenas} 
Let $b$ be a positive integer.
Let $C$ be a class of hypergraphs such that every member of $C$ has hypertreewidth at most $b$.
Then $\CQ{\Phi_C}$ has an FPRAS.
\end{thm}

 Theorem~\ref{thm:Arenas} is incomparable to Theorem~\ref{thm:submod}. 
Theorem~\ref{thm:submod} is stronger in the sense that it applies to  DCQs rather than just to CQs (indeed, we have already seen 
in Observation~\ref{obs:LBnoFPRAS}
that Theorem~\ref{thm:Arenas} cannot be extended to DCQs unless $\NP=\RP$). Theorem~\ref{thm:submod} is also stronger in the sense that it only requires bounded adaptive width instead of bounded hypertreewidth --- this gives a more general result since adaptive width strongly dominates hypertreewidth (Lemma~\ref{lem:dominate} here, from~\cite{Marx13:jacm}). 
However, Theorem~\ref{thm:Arenas} is stronger 
than Theorem~\ref{thm:submod}
in the sense that it provides an FPRAS instead of just an FPTRAS.

In this section, we prove Theorem~\ref{thm:fractionalFPRAS} which strengthens
Theorem~\ref{thm:Arenas} by bounding
\emph{fractional hypertreewidth} instead of
hypertreewidth. This is a stronger result since fractional hypertreewidth strongly dominates hypertreewidth (Lemma~\ref{lem:dominate}).
Theorem~\ref{thm:fractionalFPRAS} is still incomparable to Theorem~\ref{thm:submod} but the remaining gap now coincides with the gap between polynomial-time solvability and fixed-parameter tractability for the corresponding decision problems, see~\cite{Marx10:talg} versus~\cite{Marx13:jacm}.

\subsubsection{Fractional hypertreewidth}

\begin{defn}($H[X]$, fractional edge cover, $\fcn{H}$)
Let $H$ be a hypergraph and let $X$ be a subset of~$V(H)$. The
hypergraph \emph{induced} by $X$, denoted by $H[X]$,  
is the hypergraph with $V(H[X])=X$
and
$E(H[X])=\{e
\cap X\mid e \in E(H),  e \cap X \neq \emptyset\}$. 
A \emph{fractional edge cover} of a hypergraph $H$ is a function $\gamma:E(H)
\to [0,1]$ such that for all $v \in V(H)$, we have $\sum_{e \in E(H)\, \mid\, v \in
	e} \gamma(e) \geq 1$. The fractional edge cover number of $H$, denoted by
$\fcn{H}$, is the minimum of $\sum_{e \in E(H)} \gamma(e)$, over all fractional edge
covers $\gamma$ of $H$.
\end{defn}

We will use the following fact about fractional edge covers.
\begin{obs}\label{obs:subsets}
Given a hypergraph $H$ and subsets $B\subseteq B' \subseteq V(H)$,
we have that $\fcn{H[B]} \leq \fcn{H[B']}$.
\end{obs}
\begin{proof}
Consider a fractional edge cover $\gamma'$ of $H[B']$
such that $\fcn{H[B']} = \sum_{e'\in E(H[B'])} \gamma'(e')$.
Define a fractional edge cover $\gamma$ of $H[B]$ as follows.
Note that $E(H[B]) = \{ e' \cap B \mid e' \in E(H[B']), e'\cap B\neq \emptyset \}$.
For each $e\in E(H[B])$ let $\gamma(e) = \min\{1,\sum_{e' \in E(H[B']), e' \cap B = e} \gamma'(e')\}$.
We next show that~$\gamma$ is a fractional edge cover. First, if 
$v\in V(H[B])$ is contained in a hyperedge $e''\in E(H[B])$ where $\sum_{e' \in E(H[B']), e' \cap B = e''} \gamma'(e')\}>1$ then $\gamma(e'')=1$ so $ 
\sum_{e\in E(H[B]), v\in e} \gamma(e) \geq \gamma(e'') =1$.
Otherwise, note that for every $v\in V(H[B])$ we have
$$\sum_{e\in E(H[B]), v\in e} \gamma(e) = 
\sum_{e\in E(H[B]), v\in e} \sum_{e' \in E(H[B']), e' \cap B = e} \gamma'(e')
= 
\sum_{e'\in E(H[B']), v\in e'} \gamma'(e') 
\geq 1.
$$

Since $\gamma$ is a fractional edge cover of~$H[B]$,   $$\fcn{H[B]} \leq \sum_{e\in E(H[B])} \gamma(e) \leq 
\sum_{e'\in E(H[B'])} \gamma'(e')= \fcn{H[B']}.$$
\end{proof}

The following definition of fractional hypertreewidth, from~\cite{Grohe14:talg, Marx13:jacm},
builds on the definition of $f$-width (Definition~\ref{def:fwidth}).
 
 \begin{defn}[fractional hypertreewidth, $\fhw{H}$]
 The \emph{fractional hypertreewidth} of $H$,
 denoted by $\fhw{H}$,
 is its $f$-width where $f(X) = \fcn{H[X]}$. 
 \end{defn}

 The proof of Theorem~\ref{thm:fractionalFPRAS} follows the same approach that
 Arenas et al.~\cite{ArenasNew} used to prove Theorem~\ref{thm:Arenas} --- namely, given an input  $(\phi,\calD)$
 the following are constructed in polynomial time.
 First,
  a tree decomposition $(T,\boldB)$ of $H(\phi)$ with bounded fractional hypertree width.
 Then, using $(T,\boldB)$,
   a tree automaton~$\calT$  with the property that $|\ans{\phi,\calD}|$ is  exactly the 
 same as the  
 number of ``size~$N$'' inputs that are
 accepted by~$\calT$, where $N= |V(T)|$.
The theorem then follows  from Corollary~4.9 of~\cite{ArenasNew} which gives an FPRAS for
counting these accepted inputs.

\subsubsection{Finding a nice tree decomposition and enumerating solutions}

We will be interested in certain tree decompositions called
\emph{nice} tree decompositions \cite[Section 7.2]{ParamAlgBook}.

\begin{defn} (nice tree decomposition)
A tree decomposition $(T,\boldB)$ of a hypergraph is said to be \emph{nice} 
if the following conditions are satisfied:
\begin{itemize} 
\item the bags assigned to the root and leaf nodes of~$T$ are empty, 
\item every internal node of~$T$ has at most two children,
\item every internal node~$t$ of~$T$ with exactly two children $t_1$ and~$t_2$ has $B_t = B_{t_1} = B_{t_2}$,
\item
and every internal node~$t$ of~$T$ with exactly one child~$t_1$ has the property that the symmetric
difference of~$B_t$ and~$B_{t_1}$ has exactly one element.
\end{itemize}
\end{defn}

The following lemma builds on a tree decomposition of~\cite{Marx10:talg}.

\begin{lem}\label{lem:getniceTD}
Let $b$ be a positive integer and let  
$C$ be the set of hypergraphs with fractional treewidth at most~$b$.
There is a  polynomial-time algorithm that takes as
input a CQ $\phi\in \Phi_C$  
and a database~$\calD$ 
with $\sig(\phi) \subseteq \sig(\calD)$
and 
returns a nice tree decomposition   of~$H(\phi)$ with fractional hypertreewidth 
at most $7b^3+31b+7$
 \end{lem}
  
\begin{proof}

Let $n = \lVert\phi\rVert$ and $m= \lVert\calD\rVert$.
It is immediate from its definition that the hypergraph $H := H(\phi)$ can be constructed in~$\poly(n)$ time.
Since $H$ has fractional hypertreewidth
 at most~$b$, from \cite[Theorem 4.1]{Marx10:talg}  
there is a polynomial-time algorithm (in~$n$) for finding a  tree decomposition $(T,\boldB)$ of~$H$ with fractional hypertreewidth 
at most $c=7b^3+31b+7$. 

In $\poly(n)$ time, the tree decomposition~$(T,\boldB)$ can be turned into a nice one. The construction is standard, so we just give a sketch.
First, add a new root and new leaves with empty bags. (The new root has one child, which is the original root. Each original leaf has exactly one child, which is a new leaf.)
Then process the nodes from the root 
working down to the leaves as follows.
If $t$ has $k\geq 2$ children
$t_1,\ldots,t_k$ then this is replaced with a (nearly) complete binary tree below~$t$ with $k$ leaves $t'_1,\ldots,t'_k$ --- all nodes
in this nearly complete binary tree are given
bag~$B_t$.   
The new node~$t'_i$ is given the child~$t_i$.

Finally,
if a node~$t$ has exactly one child~$t_1$
but it is not true that the symmetric difference of~$B_t$ and~$B_{t_1}$ has exactly one element, then we replace the edge from~$t$ to~$t_1$ with a path from~$t$ to~$t_1$,  Along this path, vertices in $B_t \setminus B_{t_1}$ are dropped one-by-one then vertices in $B_{t_1} \setminus B_t$ are added.

Every bag in the nice tree decomposition is a subset of a bag in the original tree decomposition. Therefore Observation~\ref{obs:subsets}
ensures that the  nice tree composition also has fractional hypertreewidth at most~$c$.
 \end{proof}

We now extend the definitions concerning assignments from Section~\ref{sec:Intro}.
 To use these in the context
of tree decompositions, it helps to remember  (taking
$H(\phi)$ to be the hypergraph associated with
a CQ~$\phi$ from Definition~\ref{def:Hphi}) that  the vertex set of~$H(\phi)$
is defined by $V(H(\phi)) = \vars{\phi}$.

\begin{defn} (consistent assignments, $\proj$)
Let $\phi$ be a CQ and  
let $\calD$ be a database.
Suppose that $B$ and $B'$ are subsets of $\vars{\phi}$. We
say that assignments $\tau\colon B\to U(\calD)$ and $\tau' \colon B' \to U(\calD)$ are \emph{consistent} if,
for every $v\in B\cap B'$, $\tau(v) = \tau'(v)$. 
We use $\proj(\tau,B')$ to denote $\tau$'s projection onto $B'$, which is the
assignment from $B\cap B'$ to $U(\calD)$ that is consistent with $\tau$.

\end{defn}

Note that the definition of
$\proj(\tau,B')$ is consistent with
the one given in  Definition~\ref{def:ans}
for the special case where $B=\vars{\phi}$ 
and $B' = \free{\phi}$.

\begin{defn} ($\comp$)
Let $\phi$ be a CQ and  
let $\calD$ be a database.
Suppose that $B$ and $B'$ are subsets of $\vars{\phi}$
and that  $\tau\colon B\to U(\calD)$ and $\tau' \colon B' \to U(\calD)$ are  {consistent}.
We define their
composition $\comp(\tau,\tau')$ to be the unique assignment 
from $  B\cup B'$ to $ U(\calD)$ that is consistent with both~$\tau$ and~$\tau'$.\end{defn}

\begin{defn}($\proj$, as applied to sets of assignments)
Let $\phi$ be a CQ and  let $\calD$ be a database. 
If $\calL$ is a set of assignments, each from a subset of~$\vars{\phi}$ to~$U(\calD)$,
we use $\proj(\calL,B)$ to denote $\{ \proj(\tau,B) \mid \tau \in \calL\}$.
\end{defn}

Using this notation, 
the set $\ans{\phi,\calD}$  is
$\proj(\sol{\phi,\calD},\free{\phi})$.

\begin{defn} ($\sol{\phi,\calD,B}$)
Let $\phi$ be a CQ and  let $\calD$ be a database with $\sig(\phi) \subseteq \sig(\calD)$. 
Let $B$   be  a subset of~$\vars{\phi}$
A \emph{solution} 
of $(\phi,\calD,B)$ is 
an assignment $\alpha \colon B \to U(\calD)$
such that for every atom $R_i(x_{i,1},\ldots,x_{i,j})$ of~$\phi$
there is an assignment $\alpha_i \colon  \vars{\phi}\to U(\calD)$
which is consistent with~$\alpha$ and
has the property 
  that 
  $R_i^{\calD}(\alpha_i(x_{i,1}),\ldots,\alpha_i(x_{i,j}))$    is a fact in~$\calD$. 
  Let $\sol{\phi,\calD,B}$ be
the set of solutions of $(\phi,\calD,B)$.
\end{defn}

Note that a solution of $(\phi,\calD)$, as defined in
Definition~\ref{def:sol}, is the same as a solution of $(\phi,\calD,\vars{\phi})$.

The following Lemma is immediate from Theorem~3.5 of~\cite{Grohe14:talg} (though their version is stated in the language of CSPs).
\begin{lem}\label{lem:Sol} (Grohe, Marx)
Let $c$ be a positive integer.
There is a  polynomial-time algorithm that takes as
input 
\begin{itemize}
\item a CQ $\phi$  
and a database~$\calD$ 
with $\sig(\phi) \subseteq \sig(\calD)$, and
\item a subset $B \subseteq \vars{\phi}$ such that
$\fcn{H[B]} \leq c$, where $H= H(\phi)$.
\end{itemize}
The algorithm 
returns $\sol{\phi,\calD,B}$.
 \end{lem}

 \subsubsection{Tree automata}

We now give the tree automaton definitions from \cite{ArenasNew}. The tree automata that we define are
a restricted form of tree automata, corresponding to those whose trees have degree at most~$2$, but this is all that we will need.

 \begin{defn} ($\Trees_2[\Sigma]$)
 Let $\Sigma$ be a finite alphabet.
    $\Trees_2[\Sigma]$ is 
  the set of pairs $(T,\lab)$ where
  $T$ is a   rooted tree in which each vertex has at most $2$ children
  and $\lab\colon V(T) \to \Sigma$  assigns a label to each node of~$T$.
 \end{defn} 
   
  \begin{defn} (tree automaton, run, accepts, $N$-slice)
   A \emph{tree automaton}~$\calA$  
  is a tuple $(S,\Sigma,\Delta,s_0)$ where
  $S$ is a finite set of states and
  $s_0\in S$ is the initial state. $\Sigma$ is a finite alphabet. The 
   transition function~$\Delta$ is a 
   function from a subset of $S\times \Sigma$ to  
   $\{ \emptyset \} \cup S \cup (S\times S)$. 
 A \emph{run}~$\rho$ over a pair $(T,\lab)\in \Trees_2[\Sigma]$ is
 a function $\rho \colon V(T) \to S$ that assigns a state to each node of~$T$ in such a way that
 \begin{itemize}
     \item  
 for each leaf $t\in V(T)$ the transition
$(\rho(t), \lab(t)) \to \emptyset$ is in~$\Delta$,
\item
for each node $t\in V(T)$ with exactly one child~$t_1$,
$(\rho(t), \lab(t)) \to  \rho(t_1)$ is in~$\Delta$, and
\item
for each node $t\in V(T)$ with exactly two children $t_1$ and~$t_2$
(ordered left-to-right),
$(\rho(t), \lab(t)) \to  (\rho(t_1), \rho(t_2))$ is in~$\Delta$.
\end{itemize}
   The automaton~$\calA$ \emph{accepts}~$(T,\lab)$ if there is a run over~$(T,\lab)$ with $\rho(t^*) = s_0$, where $t^*$ denotes the root of~$T$.
 The $N$-slice  
 $\calL_N(\calA)$ is the set of  pairs $(T,\lab)\in  \Trees_2[\Sigma]$
 with $|V(T)| = N$
 that are accepted by~$\calA$.
  \end{defn}
 
 Given a finite alphabet~$\Sigma$,
 Arenas et al.~\cite{ArenasNew} consider the following computational problem.
 \vbox{
		\begin{description}\setlength{\itemsep}{0pt}
		\setlength{\parskip}{0pt}
		\setlength{\parsep}{0pt}   			
		        \item[\bf Name:] $\TA$ 
			    \item[\bf Input:]   A tree automaton~$\calA$   and an integer~$N$ in unary.	     			    \item[\bf Output:]  $|\calL_N(\calA)|$.
		\end{description}
	}
 
 Their Corollary~4.9 gives us the following lemma.
   \begin{lem}(Arenas, Croquevielle, Jayaram, Riveros)\label{lem:TA} 
 There is an FPRAS for $\TA$. 
 \end{lem}
 
  \subsubsection{The proof of Theorem~\ref{thm:fractionalFPRAS}}
  
  Theorem~\ref{thm:fractionalFPRAS} now follows from Lemma~\ref{lem:TA} and
  from Lemma~\ref{lem:toTA}, whose proof is inspired by the reduction of
  Arenas et al.~\cite{ArenasNew} for the case with bounded hypertreewidth.
   
\begin{lem}\label{lem:toTA}
Let $b$ be a positive integer.
Let $C$ be a class of hypergraphs such that every member of $C$ has fractional hypertreewidth at most $b$.
 There is a parsimonious reduction from
 $\CQ{\Phi_C}$  to $\TA$.\end{lem}

\begin{proof}
  
Let $(\phi,\calD)$ be an input to  $\CQ{\Phi_C}$. 
Let $n = \lVert\phi\rVert$ and $m= \lVert\calD\rVert$.

 By Lemma~\ref{lem:getniceTD}
 it takes $\poly(n,m)$ time to
 construct  a nice tree decomposition  
 $(T,\boldB)$ of $H(\phi)$
 with fractional hypertreewidth 
at most $c:=7b^3+31b+7$

For each $t\in V(T)$, let $\blanksol_t = \sol{\phi,\calD,B_t}$.
By Lemma~\ref{lem:Sol} this can be computed in $\poly(n,m)$   time.
Let   $\blanksol{}'_t = \proj(\blanksol_t,\free{\phi})$.
 Clearly, 
this can also be constructed in  
 $\poly(n,m)$ time.

 Recall that
$t^*$ is the root of $T$. Note that $\blanksol_{t^*}= \sol{\phi,\calD,\emptyset}$.
If this is empty, then  there are no answers of~$(\phi,D)$,
so the algorithm returns~$0$.
We assume from now on that $\blanksol_{t^*}$ is non-empty. In this case, it contains exactly one assignment, which is the
empty assignment~$\epsilon$.
 
 We now use $(T,\boldB)$ and $\{\blanksol_t\}$ and $\{\blanksol'_t\}$  to construct a tree automaton~$\calA = (S,\Sigma,\Delta,s_0)$.

 \begin{itemize}
     \item The state space~$S$ is defined by
     $S = \{(t,\alpha) \mid t\in V(T), \alpha \in \blanksol_t\}$.
     \item The initial state is $s_0 = (t^*,\epsilon)$.

     \item The label set~$\Sigma$ is defined by
     $\Sigma = \{(t,\beta) \mid t\in V(T), \beta \in \blanksol'_t\}$.
     
     \item  Suppose that~$t\in V(T)$ has two children $t_1$ and $t_2$.
     Recall that $B_t = B_{t_1} = B_{t_2}$. Thus, $\blanksol_t = \blanksol_{t_1} = \blanksol_{t_2}$. For each $\alpha\in \blanksol_t$ there is a   transition 
     $((t,\alpha),(t,\proj(\alpha,\free{\phi}))) \to ((t_1,\alpha),(t_2,\alpha))$. 
 
     \item Suppose that $t$ has one child~$t_1$ and $B_{t_1} \subseteq B_t$ and $B_t \setminus B_{t_1} = \{v\}$.  
     For each $\alpha\in \blanksol_t$ we have $\proj(\alpha,B_{t_1}) \in \blanksol_{t_1}$ and 
     there is a   transition 
     $((t,\alpha),(t,\proj(\alpha,\free{\phi}))) \to (t_1,\proj(\alpha,B_{t_1}))$.

     \item
   Suppose that $t$ has one child~$t_1$ and 
$B_{t} \subseteq B_{t_1}$ and $B_{t_1} \setminus B_{t} = \{v\}$.  For each 
$\alpha \in \blanksol_t$ let 
   $A_\alpha = \{ \alpha_1 \in \blanksol_{t_1} \mid \mbox{$\alpha_1$ is consistent with~$\alpha$}\}$.  
   For each $\alpha_1 \in A_\alpha$, there is a 
    transition 
     $((t,\alpha),(t,\proj(\alpha,\free{\phi}))) \to (t_1,\alpha_1)$.   
     
    \item Suppose that $t$ is a leaf so that $B_t=\emptyset$. For the empty assignment~$\epsilon$,
    there is a transition $((t,\epsilon),(t,\epsilon)) \to \emptyset$.

 \end{itemize}
 Let $N = |V(T)|$.
  The bijection from 
 $\calL_N(\calA)$  to  $\ans{\phi,\calD}$ is described as follows.
 \begin{itemize}
  \item
 Consider any element~$(T',\psi)$ of $\calL_N(\calA)$ and any run~$\rho$ that accepts~$(T',\psi)$. 
 It is immediate from the construction of~$\calA$ (and the definition of accept and run) that $T'=T$.
  
 It also follows from the construction of~$\calA$ that  the assignments~$\alpha$ in the states~$(t,\alpha)$ of~$\rho$ are all consistent. This is from the construction of the transitions, together with item (ii) in the definition of tree decomposition, which ensures that, for each variable $v\in \var(\phi)$,
 the set $\{t \in V(T) \mid v \in B_t\}$ is
 connected in~$T$.
  
  Every $v\in \var(\phi)$ appears in some set~$B_t$ (this follows from our assumption in the definition of CQs that every $v$ is in at least one atom of~$\phi$,
  together with item (i) in the definition of tree decomposition). Thus,
 composing all of the $\alpha$'s in a run  gives an assignment~$\tau\colon \var(\phi) \to U(\calD)$. 
 Consider any atom~$R$ of~$\phi$. By item~(i), the variables of~$R$ are all contained in some bag~$B_t$. 
 Let $\alpha = \proj(\tau,B_t)$.
 Since $(t,\alpha)$ is a state, $\alpha \in \blanksol_t$, so $\alpha$ satisfies~$R$ and so does~$\tau$.  
 We conclude that $\tau$ is a solution of~$(\phi,\calD)$.
 
 Therefore, $\proj(\tau,\free{\phi})$,
 which is equal to $\psi$ (the composition of
 the labels in the run~$\rho$) is an answer
 of~$(\phi,\calD)$. 
 
 \item On the other hand, for any answer~$\psi$
 of~$(\phi,\calD)$ there is a consistent   assignment~$\tau\colon \vars{\phi} \to U(\calD)$ that satisfies every atom.
 $\tau$ is a solution of~$(\phi,\calD)$.
 The run which assigns each $t\in V(T)$ state $(t,\proj(\tau,B_t))$ is
 a run accepting~$(T,\psi)$.
\end{itemize}

\end{proof}

\section{Extensions}\label{sec:extensions}
All of our algorithmic results can be extended from approximately counting answers to approximately uniformly sampling them. 
For the classes of hypergraphs $C$ for which we present algorithms for $\ECQ{\Phi_C}$ it is not hard to see that $\ECQ{\Phi_C}$ is self-partitionable (see~\cite{Dyer99:lmslns}), and thus the technique from Jerrum, Valiant and Vazirani~\cite{Jerrum86:tcs} shows that approximate sampling and approximate counting are equivalent.

Alternatively, all of our algorithms can be  
adjusted specifically to solve  the 
corresponding sampling problem.
\begin{itemize}
    \item  
A key ingredient of the algorithms presented in the proofs of Theorems~\ref{thm:mainbounded} and~\ref{thm:submod} is the framework of Dell, Lapinskas, and Meeks~\cite[Theorem 1]{DellLM20} (stated in Theorem~\ref{thm:colreduction}).
In their Theorem 2, Dell, Lapinskas, and Meeks also provide a framework for approximate sampling and this can be used  in our algorithms instead of their Theorem 1. 
\item  A key ingredient of the algorithm presented in the proof of Theorem~\ref{thm:fractionalFPRAS}
is the tree automaton algorithm of Arenas, Croquevielle, Jayaram and Riveros~\cite{ArenasNew} (stated as Lemma~\ref{lem:TA} in this work). In their Corollary~4.9, Arenas et al.\ also provide a fully-polynomial almost uniform sampler for the set of pairs that are accepted by a tree automaton and this can be used in our algorithms instead of Lemma~\ref{lem:TA}.
\end{itemize} 

Using a standard technique for approximate counting from Karp and Luby~\cite{KarpL83:focs}, our results can also be extended to counting answers of unions of (extended) conjunctive queries. This technique has been used before in other works about fixed-parameter tractability~\cite{ArvindR02:isaac}.

\bibliographystyle{plainurl}
\bibliography{CQs}

\end{document}